\documentclass[runningheads]{llncs}

\usepackage{wrapfig}
\usepackage{tikz}
\usetikzlibrary{automata,positioning}
\usepackage{graphicx}
\usepackage{color}
\usepackage{program}
\usepackage{amsmath, multirow}
\usepackage{mathtools}
\usepackage{extpfeil}
\usepackage{enumerate}
\usepackage{proof}
\usepackage[hyphens]{url}
\usepackage{hyperref}
\hypersetup{breaklinks=true}
\usepackage{breakurl}
\usepackage[normalem]{ulem}
\usepackage{xcolor, colortbl}
\usepackage{algorithm}
\usepackage{algorithmic, adjustbox, caption}
\usepackage{comment}
\usepackage{xspace}
\usepackage{listingspatched}
\usepackage{makecell}
\usepackage{pifont}
\usepackage{tabu}
\usepackage{tabularx}
\usepackage{multirow}
\usepackage{colortbl}
\usepackage{dcolumn}
\usepackage{booktabs}
\usepackage{subcaption}
\usepackage[section]{placeins}
\usepackage{marvosym}
\usepackage{centernot}

\lstdefinestyle{bpl}{
    language=C,
    basicstyle=\ttfamily,
    numberblanklines=true,
    columns=fixed,
    aboveskip=2pt,
    belowskip=1pt,
    lineskip=0pt,
    numberfirstline=true,
    xleftmargin=15pt,
    morekeywords={assert,assume,havoc,var,call,procedure,bool},
    breaklines=true,
    escapeinside={(*@}{@*)},
}
\definecolor{light-gray}{gray}{0.80}

\lstdefinestyle{lvir}{
	basicstyle=\ttfamily,
	language=,
  keywordstyle=\bfseries\ttfamily,
  keywords={do,while,if,goto,else,return}
}

\lstdefinestyle{boldinline}{
	basicstyle=\small\ttfamily\bfseries,
	language=,
  keywordstyle=,
}

\lstdefinestyle{inline}{
	basicstyle=\small\ttfamily,
	language=,
	keywordstyle=,
}

\newcommand\ie{{\it i.e.,}\xspace}
\newcommand\eg{{\it e.g.,}\xspace}

\definecolor{darkgreen}{rgb}{0.0, 0.5, 0.0}

\newcommand\GG{\ensuremath{\square}}
\newcommand\FF{\ensuremath{\lozenge}}

\definecolor{codegreen}{rgb}{0,0.6,0}
\definecolor{codegray}{rgb}{0.5,0.5,0.5}
\definecolor{codepurple}{rgb}{0.58,0,0.82}
\definecolor{backcolour}{rgb}{1.0,1.0,1.0}

\lstdefinestyle{pstyle}{
    backgroundcolor=\color{backcolour},
    commentstyle=\color{codegreen},
    keywordstyle=\color{magenta},
    numberstyle=\tiny\color{codegray},
    stringstyle=\color{codepurple},
    basicstyle=\ttfamily\footnotesize,
    breakatwhitespace=false,
    breaklines=true,
    captionpos=b,
    keepspaces=true,
    numbers=left,
    numbersep=5pt,
    showspaces=false,
    showstringspaces=false,
    showtabs=false,
    tabsize=2,
        mathescape=true,
        escapechar=|
}

\newcommand\code[1]{\lstinline|#1|}

\newcommand\aut{\mathcal{A}}

\newcommand\cfarightarrow[1]{\xtwoheadrightarrow{#1}}

\newcommand{\hurl}[1]{\href{https://#1}{\path{#1}}}

\newcommand\definetool[2]{\newcommand{#1}{\texorpdfstring{{\scshape #2}}{#2}\xspace}}
\definetool{\ultimate}     {Ultimate}
\definetool{\ultimateBB}     {UltimateBwB}
\definetool{\ltlautomizer} {Ultimate LTLAutomizer}
\definetool{\bautomizer}   {Ultimate Büchi Automizer}
\definetool{\automizer}    {Ultimate Automizer}
\definetool{\taipan}       {Ultimate Taipan}
\definetool{\kojak}        {Ultimate Kojak}
\definetool{\lassoranker}  {Ultimate LassoRanker}
\definetool{\autlib}	     {Ultimate Automata Library}
\definetool{\reqanalyzer}  {Ultimate ReqAnalyzer}
\definetool{\sautomizer}   {Automizer}
\definetool{\staipan}      {Taipan}
\definetool{\sltlautomizer}{LTLAutomizer}
\definetool{\skojak}       {Kojak}
\definetool{\cpachecker}   {CPAchecker}
\definetool{\twols}        {2LS}
\definetool{\aprove}       {AProVE}
\definetool{\kittel}       {KITTeL}
\definetool{\llvmkittel}   {LLVM2KITTeL}

\definetool{\orion}        {Orion}
\definetool{\blast}        {Blast}
\definetool{\impulse}      {Impulse}
\definetool{\slassoranker} {LassoRanker}
\definetool{\smtinterpol}  {SMTInterpol}
\definetool{\princess}	   {Princess}
\definetool{\mathsat}	     {MathSAT}
\definetool{\zzz}          {Z3}
\definetool{\cvc}	         {CVC4}
\definetool{\boogie}	     {Boogie}
\definetool{\uppaal}	     {UPPAAL}
\definetool{\assert}	     {ASSERT}
\definetool{\btc}          {BTC Embedded Platform}
\definetool{\spear}        {SpeAR}
\definetool{\smtlib}       {SMT-LIBv2}
\definetool{\benchexec}    {BenchExec}
\definetool{\BWApx}        {BWApx}
\definetool{\mcsema}       {McSema}
\definetool{\idapro}       {IDA Pro}

\newcommand\etal{{\it et al.}}

\newenvironment{itemize*}%
  {\begin{itemize}%
    \setlength{\itemsep}{0.0in}%
    \setlength{\topsep}{0.0in}%
    \setlength{\parskip}{0.0in}}%
 {\end{itemize}}
\newenvironment{enumerate*}%
  {\begin{enumerate}%
    \setlength{\itemsep}{0.0in}%
    \setlength{\topsep}{0.0in}%
    \setlength{\parskip}{0.0in}}%
  {\end{enumerate}}

\newcommand{\rewrite}[3]{\ensuremath{#1} &\ensuremath{\vdash_E}& \ensuremath{#2} &\ensuremath{\leadsto #3}}
\newcommand{\inlrewrite}[3]{\ensuremath{#1 \vdash_E #2 \leadsto #3}}
\newcommand{\weaken}[3]{\ensuremath{#1} &\ensuremath{\vdash_S}& \ensuremath{#2} &\ensuremath{\leadsto #3}}
\newcommand{\inlweaken}[3]{\ensuremath{#1} \ensuremath{\vdash_S} \ensuremath{#2} \ensuremath{\leadsto #3}}
\newcommand{\sem}[1]{\ensuremath{[\![ #1 ]\!]}}

\newcommand\cyrus[1]{\textcolor{purple}{CL: #1}}

\newcommand\Tool{{\sc DarkSea}}

\newcommand\ignore[1]{}
\newcommand\todo[1]{\textcolor{red}{[{\bf todo}: #1]}}
\newcommand\added[1]{#1}

\newcommand*\OK{\ding{52}\xspace}
\newcommand*\TOUT{\textbf{T}\xspace}
\newcommand*\OOM{\textbf{M}\xspace}
\newcommand*\NOK{\ding{55}\xspace}
\newcommand*\FAIL{\Radioactivity\xspace}
\newcommand*\UNK{{\textbf{?}\xspace}}

\newcommand{\unsound}[1]{\Lightning #1}

\newcommand{\header}[1]{\parbox{5em}{\centering #1}}
\newcommand{\rheader}[1]{\rotatebox[origin=c]{90}{#1}}
\newcommand*{\hl}{\cellcolor{black!10}}
\newcommand*{\hlg}{\cellcolor{green!20}}

\newcommand*\cirC[1]{\tikz[baseline=(char.base)]{
            \node[shape=circle,fill,inner sep=0pt, minimum size=0.5pt] (char) {\textcolor{white}{#1}};}}

\begin{document}
%
\title{Proving LTL Properties of Bitvector Programs and Decompiled Binaries (Extended)} 

%
\author{Yuandong Cyrus Liu\inst{1}
Chengbin Pang\inst{1} \and
Daniel Dietsch\inst{2} \and
Eric Koskinen\inst{1} \and
Ton-Chanh Le\inst{1} \and
Georgios Portokalidis\inst{1} \and
Jun Xu\inst{1}
}

\institute{Stevens Institute of Technology \and University of Freiburg}

%
\authorrunning{Liu {\it et al.}}
\titlerunning{Bitwise Branching for Bitvector Programs (Extended)}
\maketitle              
\setcounter{footnote}{0}
\vspace{-15pt}
\begin{abstract}
There is increasing interest in applying verification tools to programs that have bitvector operations. SMT solvers, which serve as a foundation for these tools, have thus increased support for bitvector reasoning through bit-blasting and linear arithmetic approximations.

In this paper we show that similar linear arithmetic approximation of bitvector operations can be done at the source level through transformations. Specifically, we introduce new paths that over-approximate bitvector operations with linear conditions/constraints, increasing branching but allowing us to better exploit the well-developed integer reasoning and interpolation of verification tools.  We show that, for reachability of bitvector programs, increased branching incurs negligible overhead yet, when combined with integer interpolation optimizations, enables more programs to be verified. We further show this exploitation of integer interpolation in the common case also enables competitive termination verification of bitvector programs and leads to the first effective technique for LTL verification of bitvector programs. Finally, we provide an in-depth case study of decompiled (``lifted'') binary programs, which emulate X86 execution through frequent use of bitvector operations. We present a new tool \Tool{}, the first tool capable of verifying reachability, termination, and LTL of lifted binaries. 
\end{abstract}
\section{Introduction}

There is increasing interest in using today's verification tools in domains where bitvector operations are commonplace. 
Toward this end, there has been a variety of efforts
to enable bitvector reasoning in
Satisfiability Modulo Theory (SMT) solvers, which serve as a foundation for program analysis tools.
One common strategy employed by these SMT solvers is \emph{bit-blasting}, which
translates the input bitvector formula to an equi-satisfiable propositional formula and utilizes
Boolean Satisfiability (SAT) solvers to discharge it. Another strategy is to
approximate bitvector operations with integer linear arithmetic~\cite{MathSAT:Bitvector}. 
CVC4 now employs a new approach called 
\emph{int-blasting}~\cite{CVC4:Int-blasting}, which reasons about bitvector formulas via integer nonlinear arithmetic.

Inspired by these SMT strategies, this paper explores the use of linear approximations of bitvector operations through source-level transformations\added{, toward enabling Termination/LTL verification of bitvector programs}. Our \emph{bitwise branching} introduces new conditional, linear arithmetic paths that over-approximate many but not all bitvector behaviors. These paths cover the common cases and, in the remaining cases, other paths fall back on the exact bitvector behavior. As a result, in the common case, the reasoning burden is shifted to linear arithmetic conditions/constraints, a domain more suitable to 
today's automated \added{Termination/LTL} techniques.
\added{Bitwise branching can be combined with various tools, making it an appealing general strategy. We at first chose to implement bitwise branching within \ultimate~\cite{heizmann_ultimate_nodate} source code (during the C-to-Boogie~\cite{boogiePrograms} translation) so that we could compare against unmodified Ultimate, which is already one of the more effective Termination/LTL verifiers. Furthermore, to our knowledge other tools don’t allow one to flip a switch to enable their own bit-precise analysis (i.e., CBMC’s Bitblasting or \cpachecker's FixedSizeBitVectors theory) or disable that analysis, abstracting with integers. We needed such a switch to evaluate bitwise branching. Other tools that employ non-bitprecise techniques simply report "Unknown" as soon as they encounter bit operations.}
We created rules for expressions and assignment statements, and implemented them as a source translation, \added{it is incorporated into \ultimate}.

We first examine the impact of bitwise branching on reachability and experimentally demonstrate that the translation imposes negligible overhead (from introducing additional paths to verify), yet allows existing tools to verifying more bitvector programs. Specifically, we first prepared {26} new bitvector reachability benchmarks, including examples drawn from Sean Anderson's ``BitHacks'' repository\footnote{\url{https://graphics.stanford.edu/~seander/bithacks.html}}, which use bitvector operations for various purposes. 
Without bitwise branching, \ultimate's default setting (Z3 and SMTInterpol) is only able to verify 2 of the 26 benchmarks.
We show that bitwise branching allows us to verify these benchmarks with comparable performance with existing tools
across a variety of back-end SMT solvers(\mathsat, Z3, CVC4, SMTInterpol).
We also show that bitwise branching is comparable in performance (both time and problems solved) with Z3.

The ability to use integer interpolation in the common case has far-reaching consequences, which we explore in the remainder of the paper. 
In Sec.~\ref{sec:termination} we show that,
for bitwise termination benchmarks, bitwise branching improves \ultimate\ and is competitive with other tools that support termination of bitvector programs
(\eg~\aprove, \kittel, \cpachecker).
\added{SV-COPM has a collection of benchmarks for various verification tasks, however, most SV-COMP benchmarks require little or no bitvector reasoning and they have no bitvector operations in them whatsoever. Others have bitvector operations, but those operations are not relevant to the property and existing tools abstract them away. Since the SV-COMP benchmarks do not include examples targeted to bitvector termination or bitvector LTL,}
we created new benchmarks by extending examples from the SV-COMP termination category~\cite{svcompTerm:online} and submitted our benchmarks to SV-COMP repository, as well as the the \aprove{} bitvector benchmarks \cite{aprove-bench}.

More notably, our work leads to one of the first tools for verifying temporal logic (LTL) properties of bitvector programs. To our knowledge, the only existing tool is \ultimate, and we show that bitwise branching  improves \ultimate's ability to verify LTL from merely 3 examples to a total of 59 new LTL benchmarks (out of a total of 67 benchmarks), adapted from \ultimate's LTL repository \cite{ultimate-ltl} and the BitHacks repository.

\emph{Case study: temporal verification of lifted binaries.}
In Sec.~\ref{sec:casestudy} we explore how bitwise branching can be used as part of a novel strategy for verifying decompiled (``lifted'') binaries.
Lifted binaries have lost their source 
data-types and instead emulate the 
behavior of the architecture with extensive use of bitvector operations. 
%
We developed a new tool called \Tool, built on top of our \ultimate{}-based bitwise branching, as well as \idapro{} and \mcsema.
Although these decompilation tools generate IR/C programs and today's verification tools do parse C programs, we also describe some critical translations that were needed to make the output of \mcsema{} suitable for verification (rather than re-compilation).

%
We experimentally validated our work and show that \Tool{} is the first tool for verifying temporal properties of lifted binaries. 
\Tool{} is able to prove or disprove LTL properties of 8 lifted binaries.
The most comparable alternative is \ultimate{}, which cannot prove any of them without \Tool's translations, and can only verify 6 of them without bitwise branching.

{\bf Contributions.} In summary, our contributions are: 
\vspace{-5pt}
\begin{itemize*}

\item (Sec.~\ref{sec:bitwise}) Bitwise branching for introducing paths with linear approximations.

\item (Sec.~\ref{sec:reachability})
An evaluation showing that it allows one to prove reachability of more bitvector programs, with negligible overhead.

\item (Sec.~\ref{sec:termination}) An evaluation showing competitive performance on termination, and the first effective technique
for LTL of bitvector programs.

\item (Sec.~\ref{sec:casestudy}) A case study and new tool called \Tool, the first temporal verification technique for 
decompiled (lifted) binaries.

\item New suites of bitvector benchmarks for reachability (\added{23}), termination (\added{31}), LTL (\added{41}) and lifted binaries (8).

\end{itemize*}

\noindent
We conclude with related work (Sec.~\ref{sec:relwork}).
\added{All code, proofs and benchmarks are available online\footnote{\hurl{github.com/cyruliu/darksea}}.}

\section{Motivating Examples}
\label{sec:motiv}
\begin{center}
\begin{tabular}{l|l|l}
\hline
{\bf (1) Reachability} & {\bf (2) Termination} & {\bf (3) LTL} $\varphi = \GG(\FF(n < 0))$ \\
\hline
\hline
\lstset{numbers=none}
\begin{lstlisting}[language=C,basicstyle=\tt\scriptsize]
int r, s, x;
while (x>0){
  s = x >> 31;
  x--;
  r = x + (s&(1-s));
  if (r<0) error();
}
\end{lstlisting} &
\lstset{numbers=none}
\begin{lstlisting}[language=C,basicstyle=\tt\scriptsize]
a = *; assume(a>0);
while (x>0){
  a--;
  x = x & a;
}
\end{lstlisting} &
\lstset{numbers=none}
\begin{lstlisting}[language=C,basicstyle=\tt\scriptsize,escapechar=@]
while(1) {
  n = *; x = *; y = x-1;
  while (x>0 && n>0) {
    n++;
    y = x | n;
    x = x - y;
  }
  n = -1;
}
\end{lstlisting}
\\
\hline
\code{and\_reach1.c}  & \code{and-01.c} & \code{or_loop3.c}   \\
\hline
\end{tabular}
\end{center}
\noindent
We will refer to the above bitvector programs throughout the paper.
To prove \texttt{error} unreachable in the Example {\bf (1)}, 
a verifier must be able to reason about the bitvector \lstinline|>>| and \lstinline|&| operations. Specifically, it must be able to conclude that expression \code{s&(1-s)} is always positive (so  \code{r} cannot be negative) which also depends on the earlier \code{x>>31} expression.
We will use this example to explain our work in Sec.~\ref{sec:bitwise}, and compare performance of {\ultimate} using  state-of-the-art SMT solvers, with and without bitwise branching.

\added{We will see that the key benefits of bitwise branching arise when concerned with termination and LTL.}
Example {\bf (2)} involves a simple loop, in which \code{a} is decremented, but the loop condition is on variable \code{x}, whose value is a bitvector expression over \code{a}. Today's tools for Termination of bitvector programs struggle with this example:
\aprove, \cpachecker{} and \ultimate\ report unknown and
\kittel\ and
\twols\ timeout after 900s (Apx.~\ref{apx:details:ltl:term:source}).
Critical to verifying termination of this program are (1) proving the invariant $ \code{x} > 0 \wedge \code{a} > 0$ on Line 3 within the body of the loop and (2) synthesizing a rank function. To prove the invariant, tools must show that it holds after a step of the loop's transition relation
$T = x {>} 0 \wedge a' {=} a {-} 1 \wedge x' {=} x \code{\&} a'$, which requires reasoning about the bitwise-\code{\&} operation because if we simply treat the \code{\&} as an uninterpreted function, 
$\mathcal{I} \wedge T \wedge x' {>} 0 \centernot\implies \mathcal{I'}$.

The bitwise branching strategy we describe in this paper helps the verifier infer these invariants (and later synthesize rank functions) by transforming
\begin{wrapfigure}[7]{r}[4pt]{5cm}
\vspace{-28pt}
\begin{lstlisting}[language=C,basicstyle=\tt\scriptsize,escapeinside={(*@}{@*)}]
  a = *; assume(a > 0);
  while (x > 0) {
    (*@\textcolor{blue}{ \{ x > 0 $\wedge$ a > 0 \} }@*)
    a--;
    (*@\colorbox{light-gray}{if (x >= 0 \&\& a >= 0)} @*)
    (*@\colorbox{light-gray}{then \{ x = *; assume(x <= a); \}}@*)
    (*@\colorbox{light-gray}{else \{}@*)x = x & a;(*@\colorbox{light-gray}{\}}@*)
  }
\end{lstlisting}
\end{wrapfigure}
the bitvector assignment to \code{x} into linear constraint \code{x<=a}, but only under the condition that \code{x>=0} and \code{a>=0}. 
That is, bitwise branching translates the loop in Example {\bf (2)} as depicted in the gray box to the right.
This transformation changes the transition relation of the loop body from $T$ (the original program) to $T'$:
\[
T' =  x {>} 0 \wedge a' {=} a {-} 1
          \wedge ((x {\geq} 0 \wedge a' {\geq} 0 \wedge x' {\leq} a') 
        \vee (\neg(x {\geq} 0 \wedge a' {\geq} 0) \wedge x' {=} x \code{\&} a'))\]
Importantly, when $\mathcal{I}$ holds, the else branch with the \code{\&} is infeasible, and thus we can treat the \code{\&} as an uninterpreted function and yet still prove that 
$\mathcal{I} \wedge T' \wedge x' {>} 0 \implies \mathcal{I'}$.
With the proof of $\mathcal{I}$ a tool can then move to the next step and synthesizing a ranking function $\mathcal{R}(x, a)$ that satisfies
$\mathcal{I} \wedge T' \implies \mathcal{R}(x, a) {\geq} 0 \wedge \mathcal{R}(x, a) {>} \mathcal{R}(x', a')$, namely,
$\mathcal{R}(x, a) = a$.

Bitwise branching also enables LTL verification of bitvector programs.
We examine the behavior of programs
such as Example {\bf (3)} above, with LTL property $\GG(\FF(n < 0))$.
The state of the art program verifier for LTL is \ultimate, but \ultimate{} cannot verify this program due to the bitvector operations. (\ultimate's internal overapproximation is too imprecise so it returns Unknown.) 
In Sec.~\ref{sec:termination} we show that with bitwise branching, our implementation
can prove this property of this program in 8.04s.

\medskip
\noindent
{\bf  Case study: Decompiled binary programs.}
\added{In recent years many tools have been developed for decompiling (or ``lifting'') binaries into a source code format~\cite{Brumley2011,Myreen2012,DBLP:conf/eurosys/AltinayNKRZDGNV20,dinaburg2014mcsema,de_boer_sound_2020}. The resulting code, however, has long lost the original source abstractions and instead emulates the hardware, making frequent use of bitvector operations.  These challenging programs are beyond the capabilities of existing tools for LTL verification, making them an interesting case study.}

Consider the 
(source) program  shown to the right. 
This 
program, which does
not contain any bitvector operations, is taken 
from the \ultimate{}
repository\footnote{\url{http://github.com/ultimate-pa/ultimate/blob/dev/trunk/examples/LTL/simple/PotentialMinimizeSEVPABug.c}}. 
\begin{wrapfigure}[7]{r}[6pt]{4.3cm}
\vspace*{-5ex}
\begin{lstlisting}[language=C,numbers=none]
while(1) {
  y = 1; x = *;
  while (x>0) {
    x--;
    if (x <= 1) 
      y = 0; } } }
\end{lstlisting}
\end{wrapfigure}
Some existing techniques and tools (\eg~\cite{cook_koskinen_making,ultimate-ltl}) can prove that the
LTL property 
$\GG( \code{x} > 0 \;\Rightarrow\; \FF (\code{y} = 0))$ holds. However, after the program is compiled (with \texttt{gcc}) and then disassembled and lifted (with IDPro and \mcsema), the resulting code has many bitvector operations. 
\added{The resulting lifted code is quite non-trivial (full version in  Apx.~\ref{apx:lifted-potentialminimize}) and required substantial engineering efforts just to parse and analyze with existing verifiers (see Sec.~\ref{sec:casestudy}). Let's first focus on the bitvector complexities; here is a fragment of the lifted IR (in C for readability):}
\begin{lstlisting}[language=C,basicstyle=\tt\scriptsize]
  while(true) {
    tmp_x = load i32, i32* bitcast (%x_type* @x to i32*)
    ...
    if ( ((tmp_x >> 31) == 0) & ((tmp_x == 0) ^ true) ) |\label{ln:xpos}| {
      tmp_40 = add i32 tmp_x, -1
      store i32 tmp_40, i32* bitcast (%x_type* @x to i32*)
      tmp_xp = load i32, i32* bitcast (%x_type* @x to i32*)
      tmp_42 = tmp_xp + -1;  tmp_45 = tmp_42 >> 31;
      tmp_43 = tmp_xp + -2;  tmp_44 = tmp_43 >> 31;
      if (((((((tmp_42 != 0u)&1)) & ((((((tmp_44 == 0u)&1)) ^ ((((((tmp_44 ^ tmp_45) + tmp_45)) == 2u)&1)))&1)))&1))) {|\label{ln:xone}|
         store i32 0, i32* bitcast (%y_type* @y to i32*) 
      } |\ignore{    goto block_401159}|
    } else { break; }
  }
\end{lstlisting}
Roughly, 
Line~\ref{ln:xpos} corresponds to the \code{x>0} comparison, and Line~\ref{ln:xone} corresponds to the \code{x<=1} comparison. These bitvector operations, introduced to emulate the behavior of the binary,
make it challenging for existing verification tools.

We describe a new tool \Tool{} that uses bitwise branching in the context of a decompilation toolchain involving \idapro, \mcsema{} and \ultimate. The lifting performed by tools like \mcsema{} is geared toward \emph{re}compilation rather than verification, thus foiling existing tools. In Sec.~\ref{subsec:trans} we describe translations performed by \Tool{} to tailor lifted binaries for verification.
%
In Sec.~\ref{subsec:eval2}, our experimental results show that \Tool{} is the first tool capable of proving reachability, termination and LTL of lifted binaries.

\section{Preliminaries}

Our formalization is based on Boogie programs~\cite{boogiePrograms}, denoted $P$. Our implementations parse input source C programs (or binaries recompiled to C) that may have bitvector operations. These programs are then translated into Boogie programs, in which bitvector operations are represented as uninterpreted functions. Below is an abbreviated expression syntax of $P$:
\def\uninter{\mathit{UninterpFn}}
\def\bitand{\mathit{bwAnd}}
\def\bitor{\mathit{bwOr}}
\def\bitxor{\mathit{bwXor}}
\def\shiftleft{\mathit{bwShL}}
\def\shiftright{\mathit{bwShR}}
\def\tilde{\mathit{bwCompl}}
\def\cfgor{\enspace|\enspace}
\def\cfgdef{\Coloneqq}
\def\id{\mathit{Id}}
\def\lit{\mathit{Lit}}
\def\str{\mathit{String}}
\def\prog{\mathit{Prog}}
\def\decl{\mathit{Decl}}
\def\tdecl{\mathit{TypeDecl}}
\def\cdecl{\mathit{ConstDecl}}
\def\fdecl{\mathit{FunDecl}}
\def\adecl{\mathit{AxiomDecl}}
\def\vdecl{\mathit{VarDecl}}
\def\pdecl{\mathit{ProcDecl}}
\def\idecl{\mathit{ImplDecl}}
\def\type{\mathit{Type}}
\def\typeargs{\mathit{TypeArgs}}
\def\order{\mathit{OrderSpec}}
\def\where{\mathit{WhereSpec}}
\def\spec{\mathit{Spec}}
\def\varlist{\mathit{VarList}} 
\def\varlistwhere{\mathit{VarListWhere}} 
\def\attr{\mathit{Attr}}
\def\attrarg{\mathit{AttrArg}}
\def\fsig{\mathit{FunSig}}
\def\psig{\mathit{ProcSig}}
\def\isig{\mathit{ISig}}
\def\foutp{\mathit{FOutParams}}
\def\poutp{\mathit{OutParams}}
\def\ioutp{\mathit{IOutParams}}
\def\body{\mathit{Body}}
\def\stmt{\mathit{Stmt}}
\def\lhs{\mathit{Lhs}}
\def\clhs{\mathit{CallLhs}}
\def\ifelse{\mathit{Else}}
\def\loopinv{\mathit{LoopInv}}
\def\nexpr{\mathit{NondetExpr}}
\def\expr{\mathit{Expr}}
\def\unop{\mathit{UnOp}}
\def\binop{\mathit{BinOp}}
\def\qop{\mathit{QOp}}
\def\trigattr{\mathit{TrigAttr}}
\def\plus#1{#1(\lstinline[style=boldinline]$,$~#1)^*}
\begingroup\small
\lstMakeShortInline[style=boldinline]`
\[
\begin{array}{rcl}
	e		    & \cfgdef & \lit \cfgor \id \cfgor ... \cfgor \uninter \\
	\binop		    & \cfgdef & `+` \cfgor `-` \cfgor `*` \cfgor `/` \cfgor `
                              \cfgor ..., \;\;
	\unop		     \cfgdef  `-`  \cfgor `!` \\
 	                    &  &  ... \\			
	\uninter	    & \cfgdef & \bitand \cfgor \bitor \cfgor \bitxor \cfgor
                           \shiftleft \cfgor \shiftright \cfgor \tilde \\			
\end{array}
\]
\lstDeleteShortInline`
\endgroup

Ultimate is an automaton based program analysis framework~\cite{hutchison_software_2013}, it's internal representation is boogie program with customized syntax\footnote{\hurl{https://github.com/ultimate-pa/ultimate/wiki/Boogie}}, Figure~\ref{fig:boogie_stmt} shows boogie statement syntax implemented in Ultimate, automatic abstraction techniques(such as predicate analysis~\cite{cegar} and interpolation~\cite{interpolants}) are employed on building control-flow automata.
\begin{figure}[h]
\vspace{-20pt}
  \centering
\input{fig/fig_ult_boogie_bitwise_stmt}
\vspace{-20pt}
\caption{Boogie statement syntax in Ultimate framework.}
\label{fig:boogie_stmt}
\vspace{-20pt}
\end{figure}
\begin{definition}[Control-flow automaton]
A (deterministic) \emph{control flow automaton} (CFA)~\cite{Henzinger2002} is a
tuple $\aut = \langle Q, q_0, X, s, \cfarightarrow{\,} \rangle$ where $Q$ is a
finite set of control locations and $q_0$ is the initial control location, $X$
is a finite sets of typed variables, $s$ is the loop/branch-free statement language (as defined
earlier) and $\cfarightarrow{\,}\subseteq Q \times s \times Q$ is a finite set
of labeled edges.
\end{definition}
We assume conditional branching has been transformed to
non-deterministic branching: \code{if * then \{assume($b$);$s_1$\} else \{assume(!$b$);$s_2$\}}.
As discussed later, \ultimate{} (used in our implementation)  has two modes: ``bitvector mode,'' in which these uninterpreted expressions are translated into SMT bitvector sorts and ``integer mode,''  in which they remain uninterpreted.

%

For the semantics, we assume a state space $\Sigma : \mathit{Var} \rightarrow \mathit{Val}$, mapping variables to values.
We let $\sem{e} : \Sigma \rightarrow \mathit{Val}$ and
$\sem{s} : \Sigma \rightarrow \mathcal{P}(\Sigma)$ be the semantics of expressions and statements, respectively, and
$\sem{P}$ denotes traces of $P$.

\section{Bitwise-branching}\label{sec:bitwise}

\newcommand\rr{\code{r}}
\newcommand\aaa{\code{a}}
\newcommand\bb{\code{b}}

We build our \emph{bitwise-branching} technique on the known strategy of transforming bitvector 
operations into integer approximations~\cite{MathSAT:Bitvector,CVC4:Int-blasting} but explore a new direction:
source-level transformations to introduce new conditional paths that approximate many (but not all) behaviors of a bitvector program. 
These new paths through the program have linear input conditions and linear output constraints and frequently 
cover all of the program's behavior (with respect to the goal property), but otherwise fall back on the original 
bitvector behavior when none of the input conditions hold.
We provide two sets of bitwise-branching rules:

\newcommand\ruleLabel[1]{}
\begin{figure}[!t]\centering
\subfloat[Rewriting rules for arithmetic expressions.]{\centering
  \label{fig:rewriting_rules}
  \footnotesize
\begin{tabular}{rllll}

  \rewrite{e_1=0}{e_1 \texttt{\&} e_2}{0} & \ruleLabel{R-And-0}\\

  \rewrite{(e_1=0 \vee e_1=1) \wedge e_2=1}{e_1 \texttt{\&} e_2}{e_1} & \ruleLabel{R-And-1}\\

  \rewrite{(e_1=0 \vee e_1=1) \wedge (e_2=0 \vee e_2=1)}{e_1 \texttt{\&} e_2}{e_1 \texttt{\&\&} e_2} & \ruleLabel{R-And-LOG}\\
 
  \rewrite{e_1 \geq 0 \wedge e_2 = 1}{e_1 \texttt{\&} e_2}{e_1 \texttt{\%} 2} & \ruleLabel{R-And-LBS}\\
  
  \rewrite{e_2=0}{e_1 \texttt{|} e_2}{e_1} & \ruleLabel{R-Or-0}\\
  
  \rewrite{(e_1=0 \vee e_1=1) \wedge e_2=1}{e_1 \texttt{|} e_2}{1} & \ruleLabel{R-Or-1} \\

  \rewrite{e_2=0}{e_1 \texttt{\^{}} e_2}{e_1} & \ruleLabel{R-Xor-0}\\
 
  \rewrite{e_1=e_2=0 \vee e_1=e_2=1}{e_1 \texttt{\^{}} e_2}{0} & \ruleLabel{R-Xor-Eq}\\
  
  \rewrite{(e_1=1 \wedge e_2=0) \vee (e_1=0 \wedge e_2=1)}{e_1 \texttt{\^{}} e_2}{1} & \ruleLabel{R-Xor-Neq}\\

  \rewrite{e_1 \geq 0 \wedge e_2 = \texttt{CHAR\_BIT * sizeof}(e_1) - 1}{e_1 \texttt{>>} e_2}{0} & \ruleLabel{R-RightShift-Pos}\\
  
  \rewrite{e_1 < 0 \wedge e_2 = \texttt{CHAR\_BIT * sizeof}(e_1) - 1}{e_1 \texttt{>>} e_2}{-1} & \ruleLabel{R-RightShift-Neg}

\end{tabular}
 }
\vspace{5pt}
  \hrule
\vspace{5pt}
  \subfloat[Weakening rules for relational expressions and assignments. $\texttt{op}_{le} \in \{\texttt{<,<=,==,:=}\}$, $\texttt{op}_{ge} \in \{\texttt{>,>=,==,:=}\}$, and $\texttt{op}_{eq} \in \{\texttt{==,:=}\}$]{
\label{fig:weakening_rules}
  \footnotesize
\begin{tabular}{rllll}
  \weaken{e_1 \geq 0 \wedge e_2 \geq 0}{r ~\texttt{op}_{le}~ e_1 \texttt{\&} e_2}{r \texttt{<=} e_1 ~\texttt{\&\&}~ r \texttt{<=} e_2} & \ruleLabel{W-And-Pos}\\

  \weaken{e_1 < 0 \wedge e_2 < 0}{r ~\texttt{op}_{le}~ e_1 \texttt{\&} e_2}{r \texttt{<=} e_1 ~\texttt{\&\&}~ r \texttt{<=} e_2 ~\texttt{\&\&}~ r \texttt{<} 0} & \ruleLabel{W-And-Neg}\\  
  
  \weaken{e_1 \geq 0 \wedge e_2 < 0}{r ~\texttt{op}_{eq}~ e_1 \texttt{\&} e_2}{0 \texttt{<=} r ~\texttt{\&\&}~ r \texttt{<=} e_1} & \ruleLabel{W-And-Mix}\\ 
  
  \weaken{(e_1=0 \vee e_1=1) \wedge (e_2=0 \vee e_2=1)}{(e_1 \texttt{|} e_2) \texttt{==} 0}{e_1 \texttt{==} 0 \texttt{ \&\& } e_2 \texttt{==} 0} & \ruleLabel{R-Or-LOG}\\
  
  \weaken{e_1 \geq 0 \wedge \texttt{is\_const}(e_2)}{r ~\texttt{op}_{ge}~ e_1 \texttt{|} e_2}{r \texttt{>=} e_2} & \ruleLabel{W-Or-Const}\\ 
  
  \weaken{e_1 \geq 0 \wedge e_2 \geq 0}{r ~\texttt{op}_{ge}~ e_1 \texttt{|} e_2}{r \texttt{>=} e_1 \texttt{ \&\& } r \texttt{>=} e_2} & \ruleLabel{W-Or-Pos}\\
  
  \weaken{e_1 < 0 \wedge e_2 < 0}{r ~\texttt{op}_{eq}~ e_1 \texttt{|} e_2}{r \texttt{>=} e_1 \texttt{ \&\& } r \texttt{>=} e_2 \texttt{ \&\& } r \texttt{<} 0} & \ruleLabel{W-Or-Neg}\\  
  
  \weaken{e_1 \geq 0 \wedge e_2 < 0}{r ~\texttt{op}_{eq}~ e_1 \texttt{|} e_2}{e_2 \texttt{<=} r \texttt{ \&\& } r \texttt{<} 0} & \ruleLabel{W-Or-Mix}\\ 
  
  \weaken{e_1 \geq 0 \wedge e_2 \geq 0}{r ~\texttt{op}_{ge}~ e_1 \texttt{\^{}} e_2}{r \texttt{>=} 0} & \ruleLabel{W-XOr-Pos}\\
  
  \weaken{e_1 < 0 \wedge e_2 < 0}{r ~\texttt{op}_{ge}~ e_1 \texttt{\^{}} e_2}{r \texttt{>=} 0} & \ruleLabel{W-XOr-Neg}\\  
  
  \weaken{e_1 \geq 0 \wedge e_2 < 0}{r ~\texttt{op}_{le}~ e_1 \texttt{\^{}} e_2}{r \texttt{<} 0} & \ruleLabel{W-XOr-Mix}\\
  
  \weaken{e_1 \geq 0}{r ~\texttt{op}_{le}~ {\sim}e_1}{r \texttt{<} 0} & \ruleLabel{W-Cpl-Pos}\\
  
  \weaken{e_1 < 0}{r ~\texttt{op}_{ge}~ {\sim}e_1}{r \texttt{>=} 0} & \ruleLabel{W-Cpl-Neg}\\
  
\end{tabular} 
}
\vspace{5pt}
  \hrule
\vspace{5pt}
\caption{Rewriting rules. Commutative closures omitted for brevity.}
\vspace{-20pt}
\end{figure}

{\bf 1. Rewriting rules} of the form $\inlrewrite{\mathcal{C}}{e_{bv}}{e_{int}}$ in Fig. \ref{fig:rewriting_rules}. These rules are applied to bitwise arithmetic expressions $e_{bv}$ and specify a condition $\mathcal{C}$ for which one can use integer approximate behavior $e_{int}$ of $e_{bv}$.
In other words, rewriting rule $\inlrewrite{\mathcal{C}}{e_{bv}}{e_{int}}$ can be applied only when $\mathcal{C}$ holds and a bitwise arithmetic expression $e$ in the program structurally matches its $e_{bv}$ with a substitution $\delta$. Then, $e$ will be transformed into a conditional approximation: $\mathcal{C}\delta ~?~ e_{int}\delta : e_{\added{bv}}$.
Note that, although modulo-2 is computationally more expensive, it is often more amenable to integer reasoning strategies.
  For conciseness, we omitted variants that arise from commutative re-ordering of the rules (in both Figs.~\ref{fig:rewriting_rules} and~\ref{fig:weakening_rules}).

For example, consider the bitvector arithmetic expression \code{s&(1-s)} in Example {\bf (1)} of Sec.~\ref{sec:motiv}. If we apply the rewriting rule $\inlrewrite{e_1 \geq 0 \wedge e_2 = 1}{e_1 \& e_2}{e_1 \% 2}$ with the substitution ${\code{s}/e_1, \code{1-s}/e_2}$ then the expression is transformed into \code{s>=0 \&\& (1-s)==1 ? s\%2 : (s\&(1-s))}. Since \code{s} reflects the sign bit of the positive variable \code{x}, it is always 0 and the \code{if} condition is feasible. In general, we can further replace the remaining bitwise operation in the \code{else} expression with other applicable rules. There may still be executions that fall into the final catch-all case where the bitwise operation is performed. However, as we see in the subsequent sections of this paper, these case splits are nonetheless practically significant because often the final \code{else} is infeasible.

{\bf 2. Weakening rules} of the form $\inlweaken{\mathcal{C}}{s_{bv}}{s_{int}}$ are in Fig. \ref{fig:weakening_rules}. These rules are applied to relational condition expressions (\eg~from assumptions) and assignment statements $s_{bv}$, specifying an integer condition $\mathcal{C}$ and over-approximation transition constraint $s_{int}$.
When the rule is applied to a statement (as opposed to a conditional),  replacement $s_{int}$ can be implemented as \code{assume($s_{int}$)}.
  When a weakening rule $\inlweaken{\mathcal{C}}{s_{bv}}{s_{int}}$ is applied to an assignment $s$ with substitution $\delta$, the transformed statement is $\texttt{if } \mathcal{C}\delta \texttt{ assume(} s_{int}\delta \texttt{) else } s_{\added{bv}}$.
In addition, when $s_{bv}$ of a weakening rule can be matched to the condition $\texttt{c}$ in an
$\texttt{assume(c)}$ of the original program via a substitution $\delta$, then the $\texttt{assume(c)}$ statement is transformed to
$\texttt{if } \mathcal{C}\delta \texttt{ then assume(} s_{int}\delta \texttt{) else assume(c)}$.
\begin{wrapfigure}[12]{R}[6pt]{0.5\textwidth}
  \vspace{-3em}
  \begin{tabular}{cc}
    \begin{tikzpicture}[shorten >=1pt,node distance=1.5cm, on grid,auto, semithick] 
      \node (q0) [state, initial, scale=0.6, initial where=above left, initial text = {}]{$q_0$};
      \node  (q1) [state, scale = 0.6, right=5em of q0] {$q_1$};
      \path[->] 
      (q0) 
      edge  node [auto] {\scriptsize{$x:=x\&a$}} (q1)
      (q1)  edge node[auto]  {} ++ (0.6, 0);
    \end{tikzpicture}
    \\\\
    \begin{tikzpicture}[shorten >=1pt,node distance=1.5cm,on grid,auto, semithick] 
      \node (q0) [state, initial, scale=0.6, initial text = {}]{$q_0$};
      \node  (q1a) [state, scale=0.6, below left=of q0] {$q_a$};
      \node  (q1b) [state, scale=0.6, below right=of q0] {$q_b$};
      \node  (q1) [state, scale=0.6, below left=of q1b] {$q_1$};
      
      \path[->]
      (q0) edge  [above left] node {\scriptsize{$\neg(x \geq 0 \wedge a \geq 0)$}} (q1a)
      edge [above right] node {\scriptsize{$x \geq 0 \wedge a \geq 0$}} (q1b)
      (q1a) edge [below left] node {\scriptsize{$x:=x\&a$}} (q1)          
      (q1b) edge [below right] node {\scriptsize{$assume(x\leq a)$}} (q1)
      (q1) edge node[auto] {} ++ (0.6, 0);
      
    \end{tikzpicture}
  \end{tabular}
  \caption{Weakening rules application in CFA (simplified for demonstration).}
  \label{fig:cfa_stmt}
\end{wrapfigure}
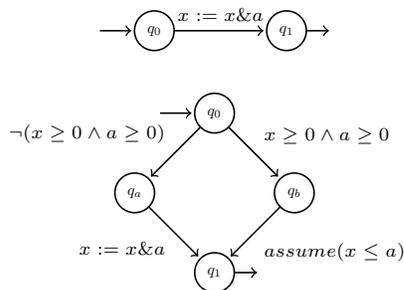
Following up Example {\bf (2)} of Sec.~\ref{sec:motiv}, consider the set of statements of given program as alphabet set in CFA, Fig.~\ref{fig:cfa_stmt} shows this rule application for $x = x \& a$ in automaton level, the top is the original automaton, the bottom is the transformed automaton after bitwise branching rule applied.

\begin{lemma}[Rule correctness]
\label{lemma:rules}
For every rule $\inlrewrite{\mathcal{C}}{e}{e'}$,
$\forall \sigma.\ \mathcal{C}(\sigma) \Rightarrow \sem{e}\sigma = \sem{e'}\sigma$.
For every $\inlweaken{\mathcal{C}}{s}{s'}$,
$\forall\sigma.\ \mathcal{C}(\sigma)\Rightarrow  \sem{s}\sigma \subseteq \sem{s'}\sigma$.
\end{lemma}
We encode each rule in Z3 script, proof details are in Appendix~\ref{apx:rules:proof}.
The rules in Fig.~\ref{fig:rewriting_rules} and Fig.~\ref{fig:weakening_rules} were developed empirically, from the 
reachability/termination/LTL benchmarks in the next sections and, especially, based on patterns found in decompiled binaries (Sec.~\ref{sec:casestudy}). We then generalized these rules to expand coverage.

\emph{Translation algorithm.}
We implemented bitwise branching via a translation algorithm, in a fork of \ultimate.
We denote our version as \ultimateBB, and will be releasing it publicly shortly.
Our translation acts on the AST of the program,
with one method
\lstinline|$T_E$ : exp -> exp| to translate expressions 
and another method
\lstinline|$T_S$ : stmt -> stmt| to translate assignment statements, each according to the set of available rules, 
\added{algorithms of $T_E$ and $T_S$ see Apx.~\ref{apx:alg}}.

In brief, when we reach a node with a bitwise operator, 
we  recursively translate the operands,
match the current operator against our collection of rules, and
apply all matching rules to construct nested if-then-else expressions/statements.
We found that, when multiple rules matched, the order did not matter much.

Let $T_E\{e\} : e$ denote the result of applying substitutions to $e$, and similar for $T_S\{s\} : s$. We lift this to a translation on a Boogie program $P$ with $T_E\{P\}:P$ and $T_S\{P\}:P$, referring to all expressions and statements in $P$, respectively.

\begin{theorem}[Soundness]\label{thm:soundness}
For every $P,T_E,T_S$,
$\sem{ P }\subseteq\sem{ T_S\{T_E\{P\}\} }$.
\end{theorem}

\begin{proof}
Induction on traces, showing equality on expression translation $T_E$ via induction on expressions/statements and then
inclusion on statement translations $T_S$.
First show that 
$T_E$ preserves traces equivalence. 
Structural induction on $e$, with base cases being constants, variables, etc. 
In the inductive case, for a bitvector operation $e_1 \otimes e_2$, assume $e_1,e_2$ has been (potentially) transformed to $e_1',e_2'$ (resp.) and that Lemma~\ref{lemma:rules} holds for each $i\in\{1,2\}$: $\forall\sigma. \sem{e_i}\sigma=\sem{e_i'}\sigma$. 
Since $\otimes$ is deterministic, $\sem{e_1'\otimes e_2'}\sigma = \sem{e_1\otimes e_2}\sigma$.
Finally, applying the transformation to $\otimes$, we show that  $\sem{T_E\{e_1'\otimes e_2'\}} = \sem{e_1'\otimes e_2'}$ again by Lemma~\ref{lemma:rules}.
Next, for each statement $s$ or relational condition $c$ step, we prove $T_S$ preserves trace inclusion:
that $\sem{s} \subseteq \sem{T_S\{s\}}$ or
that $\sem{c} \subseteq \sem{T_S\{c\}}$.
We do not recursively weaken conditional boolean expressions, which would require alternating strengthening/weakening. 
Thus, inclusion holds directly from Lemma~\ref{lemma:rules}.
\end{proof}

\newcommand\cyrulesLTL{LTLBitBench}
\newcommand\cyrulesTerm{TermBitBench}

\newcommand{\bwb}{BwB\xspace}

\section{Reachability of Bitvector Programs}
\label{sec:reachability}

We now evaluate the effectiveness of bitwise branching (BwB), as implemented in our {\ultimateBB}, toward reachability verification over our new suite of 28 bitvector programs, including those adapted from existing code snippets like the ``BitHacks'' programs
, which use bitwise operations for various tasks. 

We ran our experiments with \benchexec~\cite{benchexec} on a machine with an  AMD Ryzen 3970X 32 Core CPU with 3.7GHz and 256GB RAM running
Linux 5.4.65.
We limited CPU time to 5 minutes, memory to 8GB, and restricted each run to two cores. 
We built \ultimate 0.2.1 from source\footnote{\hurl{github.com/ultimate-pa/ultimate}, \texttt{b4afca67}, dev} and used it as baseline. 

\begin{table}[t]
\caption{Performance of \ultimate on bitvector programs,
\eg~drawn from Sean Andersen's ``Bit Hacks'' repository,
using various SMT solvers, with and without bitwise branching (\bwb).}\label{tab:reach-smt-diff} 
\begin{adjustbox}{width=\textwidth}
  \centering
  \renewcommand{\header}[1]{\parbox{4em}{\centering #1}}
\begin{tabular}{lcr@{\hspace{1em}}cr@{\hspace{1em}}cr@{\hspace{1em}}cr@{\hspace{1em}}cr@{\hspace{1em}}cr@{\hspace{1em}}cr@{\hspace{1em}}cr@{\hspace{1em}}cr@{\hspace{1em}}cr@{\hspace{1em}}cr@{\hspace{1em}}}
    \toprule
    & \multicolumn{12}{c}{Integer} & \multicolumn{8}{c}{Bitvector} \\
    \cmidrule(r){2-13}\cmidrule(l){14-21}
     & \multicolumn{2}{c}{\header{\bwb \\ \cvc}}
     & \multicolumn{2}{c}{\header{\bwb \\ MS Itp}}
     & \multicolumn{2}{c}{\header{\bwb \\ MS}}
     & \multicolumn{2}{c}{\header{\bwb \\ SItp+\zzz}}
     & \multicolumn{2}{c}{\header{SItp+\zzz}}
     & \multicolumn{2}{c}{\header{\bwb \\ \zzz}}
     & \multicolumn{2}{c}{\header{\cvc}}
     & \multicolumn{2}{c}{\header{MS Itp}}
     & \multicolumn{2}{c}{\header{MS}}
     & \multicolumn{2}{c}{\header{\zzz}} \\
     \cmidrule(r){2-3}
     \cmidrule(r){4-5}
     \cmidrule(r){6-7}
     \cmidrule(r){8-9}
     \cmidrule(r){10-11}     
     \cmidrule(r){12-13}
     \cmidrule(r){14-15}
     \cmidrule(r){16-17}
     \cmidrule(r){18-19}
     \cmidrule(r){20-21}

    {\bf Simple}                                                                                                                                                                                        \\
    logic\_cmpl.c         & \OK  & 6.10s   & \OK  & 5.28s   & \OK  & 5.43s   & \OK  & 5.68s   & \UNK & 4.90s   & \OK  & 5.91s   & \OK   & 5.90s   & \OK   & 5.65s   & \OK   & 6.86s   & \OK   & 5.67s   \\
    logic\_cmpl\_f.c      & \NOK & 6.97s   & \NOK & 6.64s   & \NOK & 6.74s   & \NOK & 5.09s   & \UNK & 6.69s   & \NOK & 6.76s   & \NOK  & 5.05s   & \NOK  & 4.74s   & \NOK  & 4.92s   & \NOK  & 4.96s   \\
    and\_loop.c           & \OK  & 7.27s   & \OK  & 5.39s   & \OK  & 6.14s   & \OK  & 5.42s   & \OK  & 5.42s   & \OK  & 5.19s   & \OK   & 5.60s   & \OK   & 6.49s   & \OK   & 5.21s   & \OK   & 5.46s   \\
    and\_loop\_f.c        & \NOK & 6.72s   & \NOK & 6.57s   & \NOK & 5.52s   & \NOK & 5.47s   & \NOK & 6.90s   & \NOK & 7.55s   & \NOK  & 5.44s   & \NOK  & 6.55s   & \NOK  & 6.92s   & \NOK  & 6.99s   \\
    logic\_or.c           & \OK  & 6.11s   & \OK  & 5.35s   & \OK  & 7.93s   & \OK  & 7.19s   & \UNK & 5.36s   & \OK  & 6.11s   & \OK   & 7.02s   & \OK   & 6.02s   & \OK   & 7.49s   & \OK   & 6.89s   \\
    logic\_or\_f.c        & \NOK & 5.35s   & \NOK & 4.91s   & \NOK & 5.52s   & \NOK & 4.58s   & \UNK & 4.85s   & \NOK & 4.93s   & \NOK  & 5.47s   & \NOK  & 5.68s   & \NOK  & 4.66s   & \NOK  & 5.22s   \\
    logic\_and.c          & \OK  & 10.91s  & \OK  & 7.27s   & \OK  & 11.66s  & \OK  & 7.30s   & \UNK & 5.64s   & \OK  & 8.83s   & \OK   & 5.89s   & \OK   & 6.88s   & \OK   & 7.04s   & \OK   & 6.80s   \\
    logic\_and\_f.c       & \NOK & 4.94s   & \NOK & 4.93s   & \NOK & 5.01s   & \NOK & 6.42s   & \UNK & 6.36s   & \NOK & 4.58s   & \NOK  & 4.98s   & \NOK  & 5.04s   & \NOK  & 4.95s   & \NOK  & 4.88s   \\
    logic\_xor.c          & \OK  & 5.99s   & \OK  & 5.54s   & \OK  & 5.94s   & \OK  & 5.06s   & \UNK & 6.64s   & \OK  & 5.69s   & \OK   & 5.77s   & \OK   & 5.23s   & \OK   & 5.30s   & \OK   & 5.25s   \\
    logic\_xor\_f.c       & \NOK & 4.90s   & \NOK & 4.93s   & \NOK & 5.19s   & \NOK & 4.64s   & \UNK & 5.16s   & \NOK & 5.16s   & \NOK  & 5.29s   & \NOK  & 5.02s   & \NOK  & 6.26s   & \NOK  & 4.87s   \\
    and\_reach1.c         & \OK  & 8.17s   & \OK  & 5.08s   & \OK  & 7.59s   & \OK  & 5.80s   & \UNK & 5.08s   & \OK  & 10.77s  & \OK   & 7.10s   & \OK   & 5.06s   & \OK   & 5.94s   & \OK   & 5.61s   \\
    and\_reach2.c         & \OK  & 6.53s   & \OK  & 5.21s   & \OK  & 6.40s   & \OK  & 5.28s   & \UNK & 6.50s   & \OK  & 8.14s   & \OK   & 5.94s   & \OK   & 5.31s   & \OK   & 6.10s   & \OK   & 7.51s   \\
    \midrule
    {\bf BitHacks}                                                                                                                                                                                      \\
    parity\_f.c           & \NOK & 6.24s   & \NOK & 5.72s   & \NOK & 5.64s   & \NOK & 5.67s   & \UNK & 5.33s   & \NOK & 6.14s   & \NOK  & 5.66s   & \NOK  & 5.32s   & \NOK  & 5.95s   & \NOK  & 5.73s   \\
    cnt-bits-set.c        & \OK  & 8.16s   & \OK  & 7.59s   & \OK  & 8.34s   & \OK  & 7.84s   & \UNK & 7.18s   & \OK  & 8.62s   & \OK   & 7.26s   & \OK   & 7.50s   & \OK   & 7.62s   & \OK   & 8.01s   \\
    cnt-bits-set\_f.c     & \NOK & 5.99s   & \NOK & 5.26s   & \NOK & 5.88s   & \NOK & 5.98s   & \UNK & 5.90s   & \NOK & 5.93s   & \NOK  & 6.56s   & \NOK  & 5.78s   & \NOK  & 5.67s   & \NOK  & 5.75s   \\
    display-bit.c         & \OK  & 5.92s   & \OK  & 7.46s   & \OK  & 5.68s   & \OK  & 7.28s   & \UNK & 6.11s   & \OK  & 6.11s   & \OK   & 7.25s   & \OK   & 5.46s   & \OK   & 6.29s   & \OK   & 5.81s   \\
    display-bit\_f.c      & \NOK & 34.16s  & \NOK & 51.33s  & \NOK & 28.61s  & \NOK & 35.44s  & \UNK & 7.60s   & \NOK & 26.64s  & \NOK  & 29.78s  & \NOK  & 44.89s  & \NOK  & 27.13s  & \NOK  & 32.37s  \\
    display-bit1.c        & \OK  & 7.08s   & \OK  & 5.83s   & \OK  & 30.98s  & \OK  & 5.81s   & \UNK & 5.43s   & \OK  & 6.47s   & \OK   & 7.69s   & \OK   & 5.42s   & \OK   & 22.28s  & \OK   & 5.65s   \\
    display-bit1\_f.c     & \NOK & 25.34s  & \NOK & 43.19s  & \NOK & 24.94s  & \NOK & 24.63s  & \UNK & 6.43s   & \NOK & 19.91s  & \NOK  & 25.27s  & \NOK  & 35.84s  & \NOK  & 20.87s  & \NOK  & 24.73s  \\
    reverse-bits1.c       & \OK  & 7.69s   & \OK  & 6.36s   & \OK  & 7.11s   & \OK  & 6.35s   & \UNK & 5.01s   & \OK  & 7.34s   & \OK   & 6.36s   & \OK   & 5.77s   & \OK   & 6.21s   & \OK   & 6.36s   \\
    reverse-bits1\_f.c    & \NOK & 7.67s   & \NOK & 7.05s   & \NOK & 7.12s   & \NOK & 7.22s   & \UNK & 6.67s   & \NOK & 7.46s   & \NOK  & 6.80s   & \NOK  & 6.57s   & \NOK  & 6.87s   & \NOK  & 6.81s   \\
    cz-bits-trailing.c    & \OK  & 6.32s   & \OK  & 5.16s   & \OK  & 5.89s   & \OK  & 6.08s   & \UNK & 6.12s   & \OK  & 5.83s   & \OK   & 5.73s   & \OK   & 6.10s   & \OK   & 5.78s   & \OK   & 6.24s   \\
    cz-bits-trailing\_f.c & \NOK & 7.11s   & \NOK & 6.72s   & \NOK & 6.68s   & \NOK & 7.22s   & \UNK & 6.54s   & \NOK & 6.77s   & \NOK  & 6.93s   & \NOK  & 7.12s   & \NOK  & 6.42s   & \NOK  & 6.79s   \\
    cnt-bits-BK1.c        & \OK  & 8.99s   & \OK  & 5.48s   & \OK  & 38.14s  & \OK  & 5.44s   & \UNK & 5.40s   & \OK  & 34.03s  & \OK   & 6.27s   & \TOUT & 300.81s & \OK   & 5.01s   & \OK   & 5.80s   \\
    cnt-bits-BK1\_f.c     & \NOK & 5.30s   & \NOK & 4.96s   & \NOK & 5.51s   & \NOK & 5.70s   & \UNK & 5.12s   & \NOK & 5.29s   & \NOK  & 5.36s   & \NOK  & 5.54s   & \NOK  & 4.99s   & \NOK  & 5.02s   \\
    cnt-bits-BK.c         & \NOK & 5.43s   & \NOK & 4.79s   & \NOK & 4.97s   & \NOK & 4.74s   & \UNK & 5.10s   & \NOK & 5.02s   & \NOK  & 5.62s   & \NOK  & 5.10s   & \NOK  & 4.81s   & \NOK  & 5.00s   \\
    cnt-bits-BK\_f.c      & \NOK & 7.70s   & \NOK & 6.92s   & \NOK & 7.41s   & \NOK & 7.23s   & \UNK & 6.60s   & \NOK & 7.06s   & \NOK  & 7.45s   & \NOK  & 7.02s   & \NOK  & 6.57s   & \NOK  & 7.50s   \\
    parity1.c             & \OK  & 6.55s   & \OK  & 19.52s  & \OK  & 6.76s   & \OK  & 6.11s   & \UNK & 24.22s  & \OK  & 6.65s   & \TOUT & 300.97s & \OK   & 80.68s  & \TOUT & 300.96s & \TOUT & 300.96s \\
    \midrule
    $\sum$Time             &      & 242.14s &      & 266.64s &      & 285.59s &      & 222.99s &      & 189.62s &      & 251.51s &       & 517.86s &       & 608.01s &       & 522.57s &       & 515.84s \\
    \bottomrule
\end{tabular}

\end{adjustbox}
\vspace{-15pt}
\end{table}

The results are summarized in Table~\ref{tab:reach-smt-diff}. 
%
Labels are {\OK} for satisfied properties, {\NOK} for violated properties with counterexamples, and {\UNK} for the unknown results where the tools could not decide. 
We also report timeouts (\TOUT), out-of-memory (\OOM), crashes (\FAIL), and highlight false positive (\unsound{\NOK}) and false negative (\unsound{\OK}) results in gray, if any.
{\ultimate} has two modes: \emph{integer} and \emph{bitvector}, each specialized to the corresponding kind of programs. In \ultimate's integer mode, overflow/underflow is accounted for with \code{assume} statements.
In its bitvector mode, {\ultimate} can utilize a variety of back-end SMT solvers with internal bitvector reasoning strategies, such as CVC4, Z3 and \mathsat{} (MS). 
By contrast, \ultimateBB{} does not use bitvector mode but instead
transforms bitvector programs (through bitwise branching) and verifies them in {\ultimate}'s integer mode using the same set of back-end SMT solvers. 
Table~\ref{tab:reach-smt-diff} shows that the performance of the integer verification with bitwise branching is comparable to the bitvector verification, despite the fact that the bitwise branching transformation may introduce many new paths.

Because {\ultimate}'s verification algorithms heavily utilize interpolation for optimizations, we also ran the experiment with interpolation enabled when possible, using \mathsat's interpolation (MS Itp, in both modes) and \smtinterpol (SItp, only in the integer mode because \smtinterpol does not support bitvectors). 
Notably, without bitwise branching, {\ultimate} can only verify 2 of 28 programs using the default setting (SItp+Z3) in its integer mode \added{while {\ultimateBB} can verify all 28 programs in the same settings}. 
Moreover, while interpolation is less effective in the bitvector mode (see MS Itp vs.\ MS), when combined with bitwise branching in the integer mode, it improves over those solvers (about 1.2x speedup, the total time for all benchmarks at the bottom row of Table~\ref{tab:reach-smt-diff}, e.g. see total time, BwB MS Itp vs BwB SItp+Z3) and has the best results (BwB SItp+Z3 column).

\section{Termination and LTL of Bitvector Programs}
\label{sec:termination}

We now evaluate bitwise branching on the main target: liveness
properties of bitvector
programs. There are few comparable tools that support bitvector  
\begin{wrapfigure}[7]{r}[6pt]{6.0cm}
\footnotesize
\begin{tabular}{|l|l|l|l|}
\hline
{\bf Tool} & {\bf BitVec}. & {\bf Term.} & {\bf LTL}  \\
\hline
\ultimate & Limited & Yes & Yes  \\
\aprove
\footnotemark
~\cite{DBLP:journals/jar/GieslABEFFHOPSS17} & Yes & Yes & No \\
\kittel
\footnotemark
~\cite{Falke2012} & Yes & Yes & No \\
\cpachecker
\footnotemark
& Limited & Yes & No \\
\twols
\footnotemark
~\cite{DBLP:conf/kbse/ChenDKSW15} & Yes & Yes & No \\
\hline
\ultimateBB & Yes & Yes & Yes  \\
\hline
\end{tabular}
\end{wrapfigure}
reasoning and these properties; the most comparable (and mature) tools are listed to the right, along with their limitations.
\footnotetext[6]{\hurl{github.com/aprove-developers/aprove-releases/releases}, \texttt{master\_2019\_09\_03}}
\footnotetext[7]{\hurl{github.com/s-falke/kittel-koat}, \texttt{c00d21f}, master}
\footnotetext[8]{\hurl{github.com/sosy-lab/cpachecker},\texttt{c2f1d8cce6}, master}
\footnotetext[9]{\hurl{github.com/diffblue/2ls}, \texttt{d35ccf73}, master}
	
\medskip\noindent
{\bf Termination.} We compare bitwise branching with the termination provers in the table to the right.
We applied these tools to
two benchmarks suites:
{\bf (i)} We first used 18 bitvector terminating programs selected from \aprove's bitvector 
\begin{wrapfigure}[18]{R}[6pt]{0.55\textwidth}
\vspace{-20pt}
\begin{tabular}{lc@{\hspace{1em}}ccccccc@{\hspace{1em}}cccccc}
      \toprule
                     & \multicolumn{6}{c}{(ii) \cyrulesTerm} &                       & \multicolumn{6}{c}{(i) AproveBench }                                                                                                                                                                                          \\
      \cmidrule(r){2-7}\cmidrule(l){9-14}
                     & \rheader{\aprove}           & \rheader{\cpachecker} & \rheader{\kittel}          & \rheader{\twols} & \rheader{\ultimate} & \rheader{\ultimateBB} &  & \rheader{\aprove} & \rheader{\cpachecker} & \rheader{\kittel} & \rheader{\twols} & \rheader{\ultimate} & \rheader{\ultimateBB} \\
      \cmidrule(r){2-7}\cmidrule(l){9-14}
      \OK            & 5                       & 1                     & 7                          & 8                & 2                   & 18          &  & 1                 & 3                     & 3                 & 14               & 2                   & 2               \\
      \unsound{\OK}  & \hl 1                       & -                     & -                          & -                & -                   & -           &  & -                 & -                     & -                 & -                & -                   & -               \\
      \NOK           & 6                       & 10                & -                          & 8            & -                   & 13              &  & -                 & -                 & -                 & -                & -               & -           \\
      \unsound{\NOK} & \hl 2                       & \hl 7                 & -                          & \hl 3            & -                   & -               &  & -                 & \hl 10                & -                 & -                & \hl 2               & \hl 6           \\
      \UNK           & 14                          & 13                    & -                          & -                & 29                  & -               &  & 10                & 3                     & -                 & 1                & 14                  & 8               \\
      \TOUT          & 3                           & -                     & 19                         & 12               & -                   & -               &  & 7                 & -                     & 10                & 2                & -                   & 1               \\
      \OOM           & -                           & -                     & -                          & -                & -                   & -               &  & -                 & -                     & -                 & 1                & -                   & 1               \\
      \FAIL          & -                           & -                     & 5                          & -                & -                   & -               &  & -                 & 2                     & 5                 & -                & -                   & -               \\
      \bottomrule
\end{tabular}


\end{wrapfigure}
benchmarks~\cite{DBLP:conf/sefm/HenselGFS16}. 
Notably, those benchmarks were designed with general bitvector arithmetic in mind so that there is only a small portion of bitvector programs in it (i.e. 18/118 or 15\%). 
{\bf (ii)} We therefore built a second set of 31 termination benchmarks, including 18 terminating programs (\OK) and 13 non-terminating programs (\NOK), called \cyrulesTerm{}
with bitvector operations including bitwise \code{$\mid$}, \code{\&}, \code{\^}, \code{<<}, \code{>>}, \code{\~}.

\emph{Results.} To the right is a table  summarizing our results (details see Apx.~\ref{apx:details:ltl:term:source}). 
For the \aprove{} benchmarks, our tool can correctly prove the termination or non-termination of 2 programs, which is less than the number of programs that can be proved by {\cpachecker} (3), {\kittel} (3), and {\twols} (14). However, for {\cyrulesTerm}, while {\ultimateBB} can prove {\em all} 31 programs, {\cpachecker}, {\kittel}, and {\twols} can only prove at most 16 programs. Moreover, while our tool was built on top of \ultimate, it outperforms {\ultimate} in proving termination and non-termination of bitwise programs.
This is because \ultimate's algorithms for synthesizing termination~\cite{hutchison_termination_2014} and non-termination proofs~\cite{DBLP:conf/tacas/LeikeH18} are not applicable to SMT formulas containing bitvectors. As a consequence, \ultimate relies on integer-based encodings of source programs together with overapproximations of bitwise operations. These results confirm that bitwise branching provides an effective means for termination of bitvector programs.
\added{Note that there are 6 false results in AproveBench for termination, they are spurious counterexamples that arise due to Ultimate’s overapproximation for unsigned integers, they do not involve branches created by our bitwise branching strategy.}

\medskip\noindent
{\bf Linear temporal logic.} We compared our tool against {\ultimate}, which is the state-of-the-art LTL prover
and the only mature LTL verifier that supports bitvector programs. 
To our knowledge, there are no available bitwise benchmarks with LTL properties so we create new benchmarks for this purpose:
(iii) New hand-crafted benchmarks called \cyrulesLTL{}
of 42 C programs with LTL properties, in which bitwise operations are heavily used in assignments, loop conditions, and branching conditions. There are 
22 programs in which the provided LTL properties are satisfied (\OK) and 20 programs in which the LTL properties are violated
(\NOK).
(iv) Benchmarks adapted from the ``BitHacks'' programs, consisting of 26 programs with LTL properties (18 satisfied and 8 violated).

\begin{wrapfigure}[12]{r}[0pt]{0.3\textwidth}
\vspace{-20pt}
\small
\renewcommand{\header}[1]{\parbox{5em}{\centering #1}}
\newcolumntype{Y}{>{\centering\arraybackslash}X}
\begin{tabularx}{0.3\textwidth}{l *{4}{Y}}
      \toprule
            & \multicolumn{2}{c}{(iv)Bithacks} & \multicolumn{2}{c}{\header{(iii)LTLBit                                           \\Bench}}                                               \\
      \cmidrule(r){2-3}\cmidrule{4-5}
            & \rheader{\ultimate}          & \rheader{w. \bwb}                 & \rheader{\ultimate} & \rheader{w. \bwb} \\
      \cmidrule(r){2-2}\cmidrule(r){3-3}\cmidrule(r){4-4}\cmidrule{5-5}
      \OK   & 3                            & 10                                & -                   & 21                \\
      \NOK  & -                            & 7                                 & -                   & 20                \\
      \UNK  & 21                           & 5                                 & 42                  & -                 \\
      \TOUT & 1                            & 1                                 & -                   & 1                 \\
      \OOM  & 1                            & 3                                 & -                   & -                 \\
      \bottomrule
\end{tabularx}

\end{wrapfigure}

%

The table to the right summarizes the result of applying {\ultimate} and {\ultimateBB} on these two bitvector
benchmarks (see Apx~\ref{apx:details:ltl:term:source} for details).
{\ultimateBB} outperforms {\ultimate}:
{\ultimateBB} can successfully verify 41 of 42 programs in \cyrulesLTL{} and 18 of 26 BitHacks programs while {\ultimate} can only handle a few of them.
Note that we have more out-of-memory results in BitHacks Benchmarks, perhaps due to memory consumption reasoning about the introduced paths.
In conclusion, bitwise branching appears to be the first effective technique for verifying LTL properties of bitvector programs.

\section{Case Study: LTL of Decompiled Binaries}
\label{sec:casestudy}

Decompiled binary executables are rife with bitvector operations, making them an interesting domain for a case study.
Many tools~\cite{dinaburg2014mcsema,Ghidra30:online,idapro,macaw,reopt,jakstab,Snowman82:online}
have been developed for decompilation. Similar to compilation, the decompilation process consists of multiple phases, beginning with disassembly.
Some techniques have emerged
for verifying low-level aspects of decompiled binaries
such as architectural semantics~\cite{Roessle2019,dasgupta_complete_2019,sail-popl20}, decompilation into logic~\cite{Myreen2007,Myreen2008,Myreen2012,de_boer_sound_2020}, and translation validation~\cite{dasgupta_scalable_2020} (discussed in Sec.~\ref{sec:relwork}).

Further along the decompilation process, other tools aim to represent a binary at a higher level of abstraction through a process called \emph{lifting}. 
A lifted binary can be represented in IR or source code, but includes only some of the source-level abstractions of the original program. 
Instead, a lifted ``program'' emulates the machine itself, with data structures that mimic the hardware (\eg~registers, flags, stack, heap, etc.) and control
that mimics the behavior of the binary. 
%

While some of the above mentioned works involve manual or semi-automated proofs of safety properties, we have not yet seen many automated techniques for verifying reachability, termination and temporal properties of those lifted binaries. 
To a large extent today's automated verification techniques have relied on source abstractions (\eg~invariants and rank functions over loop variables, structured control flow, procedure boundaries, etc.).

\subsection{Bitvector operations in lifted binaries}
Lifted binaries frequently use bitvector operations \eg~to reflect signed/unsigned comparison of variables whose type was lost in compilation.
As we show in Sec.~\ref{subsec:eval2},
lifted programs are beyond the capabilities of termination 
verification tools such as \ultimate, \cpachecker, \aprove{} or \kittel.

Returning to Example {\bf (3)},
while the source code for
the inner loop of \lstinline|PotentialMinimizeSEVPABug.c| is straight-forward (decrementing \lstinline|x| and assigning 0 to \lstinline|y| if \lstinline|x <= 1|), the corresponding
expressions in the lifted binaries involve multiple
bitvector operations:
\[\begin{array}{l}
\code{(((tmp\_42 != 0u)\&1) \& } \\
\code{((((tmp\_44 == 0u)\&1 ^ (((((tmp\_44 ^ tmp\_45) + tmp\_45) == 2u)\&1)))\&1)))\&1}
\end{array}\]
This expression simulates branch comparisons that the machine would perform on values whose type was discarded during compilation.  The source code variable \code{x} is a signed integer, but compilation has stripped its type.
During decompilation, to approximate,  lifting procedures consider these \code{tmp} variables (and all integer variables) to be unsigned. 
Meanwhile, in the binary, the condition \code{x<=0} is compiled to be a \emph{signed} comparison.
Therefore,  lifting recreates a signed comparison using the unsigned \code{tmp} variables.
Lifted binaries are good candidates for bitwise branching. For the above example,  we can use three rules: 
R-RightShift-Pos, R-And-1, R-And-Log (rule labels see Apx.~\ref{apx:rules}).

\subsection{\Tool: A toolchain for temporal verification of lifted binaries}
\label{subsec:trans}

Bitvector operations are not the only issue:
lifted binaries have several other wrinkles
that preclude them from being verified with today's tools.
We briefly discuss these issues and how we address them in a 
new toolchain called \Tool{},
the first tool capable of verifying 
reachability, termination and LTL properties of lifted
binaries.
\Tool{} is comprised of several components:
\begin{center}
\includegraphics[width=0.97\textwidth]{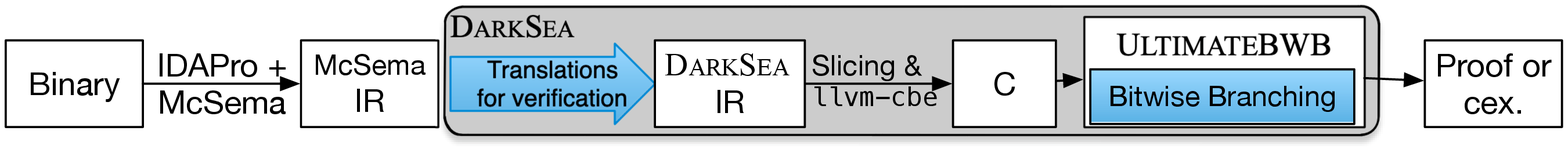}
\end{center}
\Tool\ takes as input a lifted binary (obtained from \idapro{} and \mcsema{}) in LLVM IR format, which then can be converted to C via \texttt{llvm-cbe}.

Lifting tools like \mcsema{}~\cite{DBLP:conf/eurosys/AltinayNKRZDGNV20,dinaburg2014mcsema} are often designed with the goal of \emph{re-compilation} rather than verification.
Consequently, the \mcsema{} IR, even if converted to C, cannot be analyzed by
existing tools (see Sec.~\ref{subsec:eval2}) which either crash, timeout, memout, or fail during parsing.
%
We therefore perform a series of translations discussed below to re-target the lifted binaries into a format more amenable to verification, which we then input to \ultimateBB.
%
The translations below  work with LLVM-8.0 and  consist of around 500 lines of C++ and 200 lines of \code{bash}.  We also identified and fixed several defects in \mcsema~\cite{McSema-bug1:online,McSema-bug2:online,McSema-bug3:online}.

%


\begin{enumerate*}
\item \emph{Run-time environment.} For  \emph{re}-compilation, lifting yields code that switches context between the run-time environments and the simulated code, akin to how a loader moves environment variables onto the stack. 
A first pass of \Tool{} analyzes lifted output to discover the original program's \code{main}, decouples the surrounding
context-switch code, and removes it.

\item \emph{Passing emulation state through procedures.} \mcsema\ generates lifted programs in which function arguments pass emulation state that is used for re-compilation.
We found this to make it difficult for verifiers to track state.
We thus
eliminate these arguments from every function call,
creating a single global pointer to the emulation state struct
and replacing all uses of the first argument in the function body with a use of our new pointer.
%

\item \emph{Nested structures.} 
Lifted binaries simulate hardware features (\eg~registers, arithmetic flags, FPU status flags) and, for cache efficiency,
represent them as nested structures, \eg
\code{state->general\_registers.register13.union.uint64cell}.
\Tool{} flattens these nested data structures, creating individual variables for all the innermost and separable fields, and then translates accesses to these nested structures.

\item
\emph{Property-directed slicing.}
Not all the instructions are relevant to the properties we 
aim to verify, so we further slice the program to keep 
only property-dependent code, using
DG~\cite{chalupa_mchalupadg_online}
in termination-sensitive mode.
For LTL properties, we use the atomic propositions' variables to seed our slicing criteria.

\end{enumerate*}

\noindent
More detail about these translations (and an example) can be found in Apx.~\ref{apx:fabe}.

%

\subsection{Experiments}
\label{subsec:eval2}

We evaluated whether our translations (Sec.~\ref{subsec:trans}) and bitwise branching (Sec.~\ref{sec:bitwise}) enabled tools to verify termination and LTL properties of decompiled binaries. 

\medskip
\noindent
\begin{wraptable}[16]{r}[8pt]{7cm}
  \vspace{-20pt}
  \caption{\label{tab:lift-term-overview} Termination of Lifted Binaries, with and without \Tool{} translations.}
  \vspace{-10pt}
  \centering
  \begin{tabular}{lccccccc@{\hspace{1em}}cccccc}
      \toprule
            & \multicolumn{6}{c}{\bf Raw \mcsema} &                       & \multicolumn{6}{c}{\bf \Tool{} transl.}                                                                                                                                                                                                      \\
      \cmidrule(r){2-7}\cmidrule(l){9-14}
            & \rheader{\aprove}          & \rheader{\cpachecker} & \rheader{\kittel}            & \rheader{\twols} & \rheader{\ultimate} & \rheader{{\ultimateBB}} &  & \rheader{\aprove} & \rheader{\cpachecker} & \rheader{\kittel} & \rheader{\twols} & \rheader{\ultimate} & \rheader{{\ultimateBB}} \\
            \cmidrule(r){2-7}\cmidrule(l){9-14}
      \OK   & -                          & -                     & -                            & -                & -                   & -                     &  & -                 & -                     & -                 & -                & \hlg18                  & \hlg18                    \\
      \FAIL & -                          & 18                    & -                           & -                & 3                   & -                     &  & -                 & -                     & -                & -                & -                   & -                     \\
      \OOM  & -                          & -                     & -                            & -                & -                   & 3                     &  & -                 & -                     & -                 & -                & -                   & -                     \\
      \TOUT & -                          & -                     & 18                            & -                & 15                  & 15                    &  & -                 & 18                    & 18                 & -                & -                   & -                     \\
      \UNK  & 18                         & -                     & -                            & 18               & -                   & -                     &  & 18                & -                     & -                 & 18               & -                   & -                     \\
      \bottomrule
\end{tabular}

\end{wraptable}
{\bf Termination of lifted binaries.} 
As discussed in Sec.~\ref{sec:termination}, there are several termination provers that support bitvector programs.
We thus applied those termination provers to today's lifting results on both
the raw output of \mcsema\ and
then on the output of our translation.
We used 
a standard termination benchmark (\ie 18 small, but challenging programs in literature selected from the SV-COMP \code{termination-crafted} benchmark). 
As discussed in Sec.~\ref{subsec:trans}, lifted code is more complicated than its corresponding source (\eg $>$10k vs 533 LOC in total).
Although today's termination provers can verify the source of these
programs,  they struggle to analyze 
the corresponding code lifted from the programs' binaries, as seen in the {\bf Raw \mcsema} columns 
in Table \ref{tab:lift-term-overview}
. 
(Details are given in Tables~\ref{tab:lift-term-detail-mcsema} and~\ref{tab:lift-term-detail-simplify} in Apx.~\ref{apx:fabe}.) 
We devoted genuine effort to overcome small hurdles but, fundamentally, without the \Tool\ translations, tools struggled for the following reasons:
\begin{itemize*}
\item {\aprove}: Errors in conversion from LLVM IR to internal representation.
\item {\kittel}: Parsing (from C to {\kittel}'s format via LLVM bitcode with {\llvmkittel}) succeeded, but then {\kittel} silently hung until timeout.
\item {\cpachecker}: Crashes on all benchmarks, while parsing system headers.
\item {\ultimate}: Crashes on 3 benchmarks, due to inconsistent type exceptions.
\end{itemize*}
\noindent
Table \ref{tab:lift-term-overview} also shows the verification results of those termination provers when applied to \Tool{}'s translated output (second set of columns). 

In sum, the results show that our translations benefit both {\cpachecker} and {\ultimate} (which already have sophisticated parsers), reducing crashes in analyzing  lifted code.
As highlighted in green,
\Tool\ translations enabled \ultimate\ to prove termination 
on all of the
18 lifted programs, as compared to \ultimate\ timing out on 15 of the programs without \Tool's translations.

\medskip
\noindent
{\bf LTL of lifted binaries.} We finally evaluate the effectiveness of \Tool{} in proving LTL properties of 8 lifted binaries.
In Table~\ref{tab:eval:ltl_lift} we report the LTL property and  expected verification result of each benchmark, as well as the verification time and result of {\ultimate} and {\Tool} on them (details in Apx.~\ref{apx:details:ltl:lift}).
%
Green cells use slightly different settings (enabled SBE, there are various setting strategies in \ultimate framework~\cite{heizmann_ultimate_nodate}).
\Tool's translations eliminate unsoundness results that come from applying \ultimate\ directly to \mcsema\ IR (detailed table in Apx.~\ref{apx:details:ltl:lift}).
\begin{table}[!h]
    \vspace{-20pt}
    \centering
    \caption{\ultimate\ vs. \Tool\ on lifted programs with LTL properties.}\label{tab:eval:ltl_lift}
    \small
\begin{tabular}{llcc@{\hspace{1em}}rcrc}
    \toprule
    \cmidrule(r){5-8}
                            &                                              &      &  & \multicolumn{2}{c}{\ultimate}  & \multicolumn{2}{c}{\Tool}                         \\

    Benchmark               & Property                                     & Exp. &  & Time                          & Result            & Time       & Result\\
    \cmidrule(r){1-1}\cmidrule(lr){2-2}\cmidrule(l){3-3} \cmidrule(r){5-6}\cmidrule(l){7-8} 
    01-exsec2.s.c           & $\lozenge (\square x=1)$                     & \OK  &  & 4.45s                         & \FAIL                     & 11.23s     & \OK      \\
    01-exsec2.s.f.c.c       & $\lozenge (\square x \neq 1)$                & \NOK &  & 6.31s                         & \FAIL                    & 10.36s     & \NOK     \\
    SEVPA\_gccO0.s.c        & $\square  ( x > 0 \Rightarrow \lozenge y=0)$ & \OK  &  & 6.31s                         & \FAIL                     & \hlg22.92s & \hlg\OK  \\
    SEVPA\_gccO0.s.f.c      & $\square  ( x > 0 \Rightarrow \lozenge y=2)$ & \NOK &  & 5.16s                         & \UNK                     & \hlg14.92s & \hlg\NOK \\
    acqrel.simplify.s.c     & $\square (x=0 \Rightarrow \lozenge y=0)$     & \OK  &  & 5.17s                         & \FAIL                     & 9.00s      & \OK      \\
    acqrel.simplify.s.f.c.c & $\square (x=0 \Rightarrow \lozenge y=1)$     & \NOK &  & 6.06s                         & \FAIL                     & 17.60s     & \NOK     \\
    exsec2.simplify.s.c     & $\square \lozenge x=1$                       & \OK  &  & 4.92s                         & \FAIL                     & 5.60s      & \OK      \\
    exsec2.simplify.s.f.c.c & $\square \lozenge x\neq 1$                   & \NOK &  & 4.55s                         & \FAIL                     & 6.28s      & \NOK     \\
    \bottomrule
\end{tabular}

    \vspace{-20pt}
\end{table}

In summary, we have shown that \Tool{} can verify reachability, termination and LTL properties of lifted binaries. To our knowledge, \Tool{} is the first to do so.

\section{Related Work}
\label{sec:relwork}

\emph{Bitvector reasoning.}
%
%
Many works support bitvector reasoning in SMT solvers (\eg~\cite{wintersteiger_efficiently_2013}).
Kroening \etal~\cite{hermanns_approximating_2006} 
perform predicate image over-approximation.
%
Niemetz \etal~\cite{niemetz2019bitwidthindependent} propose a translation from bitvector formulas with parametric bit-width to formulas in a logic supported by SMT solvers, making SMT-based procedures available for variant-size bitvector formulas.

He and Rakamari\'{c}~\cite{chang_counterexample-guided_2017} 
build on spurious counterexamples from overapproximations of bitvector operations.
%
Mattsen \etal~\cite{mattsen_non-convex_2015} use a BDD-based abstract domain for indirect jump reasoning.
Bryant \etal~\cite{grumberg_deciding_2007} 
iterative construct an abstraction of a bit vector formula.

Other works have targeted reasoning about \emph{termination}  of bitvector programs.
Cook \etal~\cite{CookBV2010} use Presburger arithmetic for
representing rank functions.
Chen \etal~\cite{Chen2018} employ lexicographic rank function synthesis for bit precision and rely on the bit-precision of an underlying SMT solver.
%
Falke \etal~\cite{Falke2012} propose an approach, implemented in \kittel, which 
derives linear approximations of bitvector operations using some rules similar to our bitwise-branching rules for expressions. However, Falke \etal~create a large disjunction of cases which puts a large burden on the solver. By contrast, our bitwise-branching creates multiple verification paths, but solver queries for most of them can be avoided through integer interpolation.
As we show in Sec.~\ref{sec:termination}, our {\ultimateBB} was able to solve 33/49 benchmarks, where as \kittel\ solved only 10.
Moreover, \kittel\ does not support LTL properties and crashes on lifted binaries.
%
%



\emph{Tools for disassembly and decompilation.} 
Jakstab~\cite{kinder2010precise} focuses on accurate control flow reconstruction in the \emph{disassembly} process.
BAP~\cite{Brumley2011} performs static disassembly of stripped binaries. 
Angr~\cite{shoshitaishvili2016sok} includes symbolic execution and value-set analysis used especially for control flow reconstruction. IDA Pro~\cite{idapro} (used in \Tool) demonstrated high accuracy and uses value-set-analysis. 
%
%
Hex-Rays Decompiler~\cite{HexRaysD5:online}, Ghidra~\cite{Ghidra30:online}, and Snowman~\cite{Snowman82:online} 
further de-compile disassembled output to higher level representations such as LLVM IR or C code. 

\emph{Verifying binaries.} 
Some works focus on the low-level aspects of the binary and aim at precise de-compilation. Roessle \etal~\cite{Roessle2019} de-compile x86-64 into a big step semantics.
Earlier, others performed ``decompilation-into-logic'' (DiL)~\cite{Myreen2007,Myreen2008,Myreen2012}, translating assembly code into logic.
While DiL provides a rich environment for precise reasoning about fine-grained instruction-level details, it incurs high complexity for reasoning about more coarse-grained properties such as reachability, termination, and temporal logic.
In more recent work, 
Verbeek \etal~\cite{de_boer_sound_2020} use the semantics of Roessle \etal~\cite{Roessle2019} and describe techniques 
to decompile into re-compilable code.


Others focus on verifying the decompilation/lifting process itself.
%
Dasgupta \etal~\cite{dasgupta_scalable_2020} describe a translation validation on x86-64 instructions that employs their semantics for x86-64 (Dasgupta \etal~\cite{dasgupta_complete_2019}). 
%
Metere \etal~\cite{metere_sound_2017} use HOL4 to verify a translation from ARMv8 to BAP.
Hendrix \etal~\cite{hendrix_towards_nodate} discuss their ongoing work on verifying the translation performed by their lifting tool \code{reopt}.
%
Numerous other works (\eg~Sail~\cite{sail-popl20}) provide formal semantics of ISAs.

\section{Conclusion}
We have shown that a source-level translation to approximate bitvector operations leads to tools that are competitive to the state-of-the-art in reachability and termination of bitvector programs.
We show that bitwise branching incurs negligible overhead, yet enables more programs to be verified. Notably, we showed that this approach leads to the first effective technique for verifying LTL of bitvector programs and, to our knowledge, the first technique for verifying reachability, termination and LTL of lifted binary programs.  

\emph{Acknowledgments.}
This work is supported by ONR Grant \#N00014-17-1-2787.

{\footnotesize
\bibliographystyle{splncs04}
\bibliography{biblio}

\begin{thebibliography}{10}
\providecommand{\url}[1]{\texttt{#1}}
\providecommand{\urlprefix}{URL }
\providecommand{\doi}[1]{https://doi.org/#1}

\bibitem{aprove-bench}
{AProVE}. \hurl{aprove.informatik.rwth-aachen.de/eval/Bitvectors/}

\bibitem{gcc-bug2:online}
Gcc bug2. \hurl{gcc.gnu.org/bugzilla/show_bug.cgi?id=42952}

\bibitem{HexRaysD5:online}
Hex-rays decompiler. \hurl{www.hex-rays.com/products/decompiler/}

\bibitem{McSema-bug1:online}
\mcsema{}~ jump table bug. \hurl{github.com/lifting-bits/mcsema/issues/558}

\bibitem{McSema-bug3:online}
\mcsema{} bug, missing data cross reference due to resetting ida's analysis
  flag. \hurl{github.com/lifting-bits/mcsema/issues/561}

\bibitem{McSema-bug2:online}
\mcsema{} var.~bug. \hurl{github.com/lifting-bits/mcsema/issues/566}

\bibitem{svcompTerm:online}
{SV-COMP Termination Benchmarks}.
  \url{github.com/sosy-lab/sv-benchmarks/tree/master/c/termination-crafted}

\bibitem{ultimate-ltl}
Ultimate's {LTL} benchmarks.
  \url{github.com/ultimate-pa/ultimate/tree/dev/trunk/examples/LTL/}

\bibitem{Ghidra30:online}
Agency, N.S.: Ghidra. \url{www.nsa.gov/resources/everyone/ghidra/}

\bibitem{DBLP:conf/eurosys/AltinayNKRZDGNV20}
Altinay, A., Nash, J., Kroes, T., Rajasekaran, P., Zhou, D., Dabrowski, A.,
  Gens, D., Na, Y., Volckaert, S., Giuffrida, C., Bos, H., Franz, M.: Binrec:
  dynamic binary lifting and recompilation. In: EuroSys. pp. 36:1--36:16 (2020)

\bibitem{sail-popl20}
Armstrong, A., Bauereiss, T., Campbell, B., Reid, A., Gray, K.E., Norton, R.M.,
  Mundkur, P., Wassell, M., French, J., Pulte, C., Flur, S., Stark, I.,
  Krishnaswami, N., Sewell, P.: {ISA} semantics for {ARMv8-a}, {RISC-v}, and
  {CHERI-MIPS}. Proc. ACM Program. Lang.  \textbf{3}(POPL) (Jan 2019)

\bibitem{boogiePrograms}
Barnett, M., Chang, B.Y.E., DeLine, R., Jacobs, B., Leino, K.R.M.: Boogie: A
  modular reusable verifier for object-oriented programs. In: International
  Symposium on Formal Methods for Components and Objects. pp. 364--387 (2005)

\bibitem{benchexec}
Beyer, D., L{\"{o}}we, S., Wendler, P.: Reliable benchmarking: requirements and
  solutions. Int. J. Softw. Tools Technol. Transf.  \textbf{21}(1),  1--29
  (2019)

\bibitem{MathSAT:Bitvector}
Bozzano, M., Bruttomesso, R., Cimatti, A., Franz{\'{e}}n, A., Hanna, Z.,
  Khasidashvili, Z., Palti, A., Sebastiani, R.: {Encoding {RTL} Constructs for
  MathSAT: a Preliminary Report}. Electron. Notes Theor. Comput. Sci.
  \textbf{144}(2),  3--14 (2006)

\bibitem{Brumley2011}
Brumley, D., Jager, I., Avgerinos, T., Schwartz, E.J.: {BAP:} {A} binary
  analysis platform. In: Computer Aided Verification - 23rd International
  Conference, {CAV} 2011, Snowbird, UT, USA, July 14-20, 2011. Proceedings. pp.
  463--469 (2011)

\bibitem{grumberg_deciding_2007}
Bryant, R.E., Kroening, D., Ouaknine, J., Seshia, S.A., Strichman, O., Brady,
  B.: Deciding {Bit}-{Vector} {Arithmetic} with {Abstraction}. In: Tools and
  {Algorithms} for the {Construction} and {Analysis} of {Systems}, vol.~4424,
  pp. 358--372 (2007)

\bibitem{chalupa_mchalupadg_online}
Chalupa, M.: mchalupa/dg. \hurl{github.com/mchalupa/dg} (Jan 2021)

\bibitem{DBLP:conf/kbse/ChenDKSW15}
Chen, H., David, C., Kroening, D., Schrammel, P., Wachter, B.: Synthesising
  interprocedural bit-precise termination proofs {(T)}. In: {ASE}. pp. 53--64
  (2015)

\bibitem{Chen2018}
Chen, H.Y., David, C., Kroening, D., Schrammel, P., Wachter, B.: Bit-{Precise}
  {Procedure}-{Modular} {Termination} {Analysis}. ACM Transactions on
  Programming Languages and Systems  \textbf{40},  1--38 (Jan 2018)

\bibitem{cegar}
Clarke, E., Grumberg, O., Jha, S., Lu, Y., Veith, H.: Counterexample-guided
  abstraction refinement. In: International Conference on Computer Aided
  Verification. pp. 154--169. Springer (2000)

\bibitem{cook_koskinen_making}
Cook, B., Koskinen, E.: Making prophecies with decision predicates. In:
  Proceedings of the 38th Annual ACM SIGPLAN-SIGACT Symposium on Principles of
  Programming Languages. p. 399–410. POPL '11 (2011)

\bibitem{CookBV2010}
Cook, B., Kroening, D., Rummer, P., Wintersteiger, C.M.: Ranking function
  synthesis for bit-vector relations. In: Esparza, J., Majumdar, R. (eds.)
  Tools and Algorithms for the Construction and Analysis of Systems. pp.
  236--250 (2010)

\bibitem{dasgupta_scalable_2020}
Dasgupta, S., Dinesh, S., Venkatesh, D., Adve, V.S., Fletcher, C.W.: Scalable
  validation of binary lifters. In: Proceedings of the 41st {ACM} {SIGPLAN}
  {Conference} on {Programming} {Language} {Design} and {Implementation}. pp.
  655--671 (Jun 2020)

\bibitem{dasgupta_complete_2019}
Dasgupta, S., Park, D., Kasampalis, T., Adve, V.S., Roşu, G.: A {Complete}
  {Formal} {Semantics} of x86-64 {User}-{Level} {Instruction} {Set}
  {Architecture} p.~16 (2019)

\bibitem{Snowman82:online}
Derevenets, Y.: Snowman. \hurl{derevenets.com/}

\bibitem{dillig2011precise}
Dillig, I., Dillig, T., Aiken, A., Sagiv, M.: Precise and compact modular
  procedure summaries for heap manipulating programs. {PLDI}  \textbf{46}(6),
  567--577 (2011)

\bibitem{dinaburg2014mcsema}
Dinaburg, A., Ruef, A.: Mcsema: Static translation of x86 instructions to llvm.
  In: ReCon 2014 Conference, Montreal, Canada (2014)

\bibitem{Falke2012}
Falke, S., Kapur, D., Sinz, C.: Termination {Analysis} of {Imperative}
  {Programs} {Using} {Bitvector} {Arithmetic}. In: Verified {Software}:
  {Theories}, {Tools}, {Experiments}, vol.~7152, pp. 261--277 (2012)

\bibitem{macaw}
Galois, I.: Macaw, \hurl{github.com/GaloisInc/macaw}

\bibitem{reopt}
Galois, I.: Reopt vcg, \hurl{github.com/GaloisInc/reopt-vcg}

\bibitem{DBLP:journals/jar/GieslABEFFHOPSS17}
Giesl, J., Aschermann, C., Brockschmidt, M., Emmes, F., Frohn, F., Fuhs, C.,
  Hensel, J., Otto, C., Pl{\"{u}}cker, M., Schneider{-}Kamp, P., Str{\"{o}}der,
  T., Swiderski, S., Thiemann, R.: Analyzing program termination and complexity
  automatically with aprove. J. Autom. Reason.  \textbf{58}(1),  3--31 (2017)

\bibitem{chang_counterexample-guided_2017}
He, S., Rakamarić, Z.: Counterexample-{Guided} {Bit}-{Precision} {Selection}.
  In: Programming {Languages} and {Systems}, vol. 10695, pp. 534--553 (2017)

\bibitem{heizmann_ultimate_nodate}
Heizmann, M., Christ, J., Dietsch, D., Hoenicke, J., Lindenmann, M., Musa, B.,
  Schilling, C., Wissert, S., Podelski, A.: Ultimate program analysis framework
  p.~1

\bibitem{nestedInterpolants}
Heizmann, M., Hoenicke, J., Podelski, A.: Nested interpolants. ACM Sigplan
  Notices  \textbf{45}(1),  471--482 (2010)

\bibitem{hutchison_software_2013}
Heizmann, M., Hoenicke, J., Podelski, A.: Software {Model} {Checking} for
  {People} {Who} {Love} {Automata}. In: Computer {Aided} {Verification},
  vol.~8044, pp. 36--52 (2013)

\bibitem{hutchison_termination_2014}
Heizmann, M., Hoenicke, J., Podelski, A.: Termination {Analysis} by {Learning}
  {Terminating} {Programs}. In: Computer {Aided} {Verification}, vol.~8559, pp.
  797--813 (2014)

\bibitem{hendrix_towards_nodate}
Hendrix, J., Wei, G., Winwood, S.: Towards {Verified} {Binary} {Raising} p.~4

\bibitem{DBLP:conf/sefm/HenselGFS16}
Hensel, J., Giesl, J., Frohn, F., Str{\"{o}}der, T.: Proving termination of
  programs with bitvector arithmetic by symbolic execution. In: {SEFM}.
  vol.~9763 (2016)

\bibitem{Henzinger2002}
Henzinger, T.A., Jhala, R., Majumdar, R., Necula, G.C., Sutre, G., Weimer, W.:
  Temporal-safety proofs for systems code. In: Computer Aided Verification,
  14th International Conference, {CAV} 2002,Copenhagen, Denmark, July 27-31,
  2002, Proceedings. pp. 526--538 (2002)

\bibitem{jakstab}
Kinder, J.: Jakstab, \url{http://www.jakstab.org/}

\bibitem{kinder2010precise}
Kinder, J., Veith, H.: Precise static analysis of untrusted driver binaries.
  In: Formal Methods in Computer Aided Design. pp. 43--50. IEEE (2010)

\bibitem{hermanns_approximating_2006}
Kroening, D., Sharygina, N.: Approximating {Predicate} {Images} for
  {Bit}-{Vector} {Logic}. In: {TACAS}. Lecture Notes in Computer Science,
  vol.~3920, pp. 242--256 (2006)

\bibitem{DBLP:conf/tacas/LeikeH18}
Leike, J., Heizmann, M.: Geometric nontermination arguments. In: {TACAS} {(2)}.
  Lecture Notes in Computer Science, vol. 10806, pp. 266--283 (2018)

\bibitem{mattsen_non-convex_2015}
Mattsen, S., Wichmann, A., Schupp, S.: A non-convex abstract domain for the
  value analysis of binaries. In: {SANER}. pp. 271--280 (2015)

\bibitem{interpolants}
McMillan, K.L.: Applications of craig interpolants in model checking. In:
  International Conference on Tools and Algorithms for the Construction and
  Analysis of Systems. pp. 1--12. Springer (2005)

\bibitem{metere_sound_2017}
Metere, R., Lindner, A., Guanciale, R.: Sound {Transpilation} from {Binary} to
  {Machine}-{Independent} {Code}. arXiv:1807.10664 [cs]  \textbf{10623},
  197--214 (2017)

\bibitem{Myreen2007}
Myreen, M.O., Gordon, M.J.C.: Hoare logic for realistically modelled machine
  code. In: {TACAS}. Lecture Notes in Computer Science, vol.~4424, pp. 568--582
  (2007)

\bibitem{Myreen2008}
Myreen, M.O., Gordon, M.J.C., Slind, K.: Machine-code verification for multiple
  architectures - an application of decompilation into logic. In: Formal
  Methods in Computer-Aided Design, {FMCAD} 2008. pp.~1--8 (2008)

\bibitem{Myreen2012}
Myreen, M.O., Gordon, M.J.C., Slind, K.: Decompilation into logic - improved.
  In: Formal Methods in Computer-Aided Design, {FMCAD} 2012, Cambridge, UK,
  October 22-25, 2012. pp. 78--81 (2012)

\bibitem{niemetz2019bitwidthindependent}
Niemetz, A., Preiner, M., Reynolds, A., Zohar, Y., Barrett, C., Tinelli, C.:
  Towards bit-width-independent proofs in smt solvers (2019)

\bibitem{gcc-bug3:online}
Regehr, J.: Gcc bug3. \hurl{gcc.gnu.org/bugzilla/show_bug.cgi?id=42721}

\bibitem{Roessle2019}
Roessle, I., Verbeek, F., Ravindran, B.: Formally verified big step semantics
  out of x86-64 binaries. In: Proceedings of the 8th {ACM} {SIGPLAN}
  International Conference on Certified Programs and Proofs (2019)

\bibitem{idapro}
SA, H.R.: Ida pro, \hurl{www.hex-rays.com/products/ida/}

\bibitem{shoshitaishvili2016sok}
Shoshitaishvili, Y., Wang, R., Salls, C., Stephens, N., Polino, M., Dutcher,
  A., Grosen, J., Feng, S., Hauser, C., Kruegel, C., et~al.: Sok:(state of) the
  art of war: Offensive techniques in binary analysis. In: 2016 IEEE Symposium
  on S\&P (2016)

\bibitem{de_boer_sound_2020}
Verbeek, F., Olivier, P., Ravindran, B.: Sound {C} {Code} {Decompilation} for a
  {Subset} of x86-64 {Binaries}. In: Software {Engineering} and {Formal}
  {Methods}, vol. 12310, pp. 247--264 (2020), series Title: Lecture Notes in
  Computer Science

\bibitem{wintersteiger_efficiently_2013}
Wintersteiger, C.M., Hamadi, Y., de~Moura, L.: Efficiently solving quantified
  bit-vector formulas. Formal Methods in System Design  \textbf{42},  3--23
  (Feb 2013)

\bibitem{summaries}
Yorsh, G., Yahav, E., Chandra, S.: Generating precise and concise procedure
  summaries. In: Proceedings of the 35th annual ACM SIGPLAN-SIGACT symposium on
  Principles of programming languages. pp. 221--234 (2008)

\bibitem{CVC4:Int-blasting}
Zohar, Y., Irfan, A., Mann, M., Niemetz, A., Notzli, A., Preiner, M., Reynolds,
  A., Barrett, C., Tinelli, C.: {Bit-Precise Reasoning via Int-Blasting} (2021)

\end{thebibliography}
}

\vfill
\pagebreak
\appendix
\section{Full Lifted Code for \lstinline@PotentialMinimizeSEVPABug@ (Sec.~\ref{sec:motiv})}
\label{apx:lifted-potentialminimize}


\begin{lstlisting}[language=C, basicstyle=\scriptsize]
//@ ltl invariant positive: ([] ( AP(x > 0) ==> <>AP(y==0)));

/* Provide Declarations */
#include <stdarg.h>
#include <setjmp.h>
#include <limits.h>
#include <stdint.h>
#include <math.h>

/* Global Declarations */

/* Types Declarations */
struct l_struct_x_type;
struct l_struct_y_type;
struct l_struct_union_OC_anon;
struct l_struct_struct_OC_ArchState;
struct l_struct_struct_OC_uint64v8_t;
struct l_struct_union_OC_vec512_t;
struct l_struct_union_OC_VectorReg;
struct l_struct_struct_OC_ArithFlags;
struct l_struct_union_OC_SegmentSelector;
struct l_struct_struct_OC_Segments;
struct l_struct_struct_OC_Reg;
struct l_struct_struct_OC_AddressSpace;
struct l_struct_struct_OC_GPR;
struct l_struct_struct_OC_anon_OC_3;
struct l_struct_struct_OC_X87Stack;
struct l_struct_struct_OC_uint64v1_t;
struct l_struct_union_OC_vec64_t;
struct l_struct_struct_OC_anon_OC_4;
struct l_struct_struct_OC_MMX;
struct l_struct_struct_OC_FPUStatusFlags;
struct l_struct_union_OC_FPUAbridgedTagWord;
struct l_struct_union_OC_FPUControlStatus;
struct l_struct_struct_OC_float80_t;
struct l_struct_union_OC_anon_OC_11;
struct l_struct_struct_OC_FPUStackElem;
struct l_struct_struct_OC_uint128v1_t;
struct l_struct_union_OC_vec128_t;
struct l_struct_struct_OC_FpuFXSAVE;

/* Types Definitions */
struct l_array_4_ureplace_u8int {
   int array[4];
};
struct l_struct_x_type {
  struct l_array_4_ureplace_u8int field0;
} __attribute__ ((packed));
struct l_array_8_ureplace_u8int {
   int array[8];
};
struct l_struct_y_type {
  struct l_array_8_ureplace_u8int field0;
} __attribute__ ((packed));
struct l_struct_union_OC_anon {
   int field0;
};
struct l_struct_struct_OC_ArchState {
   int field0;
   int field1;
  struct l_struct_union_OC_anon field2;
};
struct l_array_8_ureplace_u64int {
   int array[8];
};
struct l_struct_struct_OC_uint64v8_t {
  struct l_array_8_ureplace_u64int field0;
};
struct l_struct_union_OC_vec512_t {
  struct l_struct_struct_OC_uint64v8_t field0;
};
struct l_struct_union_OC_VectorReg {
  struct l_struct_union_OC_vec512_t field0;
};
struct l_array_32_struct_AC_l_struct_union_OC_VectorReg {
  struct l_struct_union_OC_VectorReg array[32];
};
struct l_struct_struct_OC_ArithFlags {
   int field0;
   int field1;
   int field2;
   int field3;
   int field4;
   int field5;
   int field6;
   int field7;
   int field8;
   int field9;
   int field10;
   int field11;
   int field12;
   int field13;
   int field14;
   int field15;
};
struct l_struct_union_OC_SegmentSelector {
   short field0;
};
struct l_struct_struct_OC_Segments {
   short field0;
  struct l_struct_union_OC_SegmentSelector field1;
   short field2;
  struct l_struct_union_OC_SegmentSelector field3;
   short field4;
  struct l_struct_union_OC_SegmentSelector field5;
   short field6;
  struct l_struct_union_OC_SegmentSelector field7;
   short field8;
  struct l_struct_union_OC_SegmentSelector field9;
   short field10;
  struct l_struct_union_OC_SegmentSelector field11;
};
struct l_struct_struct_OC_Reg {
  struct l_struct_union_OC_anon field0;
};
struct l_struct_struct_OC_AddressSpace {
   int field0;
  struct l_struct_struct_OC_Reg field1;
   int field2;
  struct l_struct_struct_OC_Reg field3;
   int field4;
  struct l_struct_struct_OC_Reg field5;
   int field6;
  struct l_struct_struct_OC_Reg field7;
   int field8;
  struct l_struct_struct_OC_Reg field9;
   int field10;
  struct l_struct_struct_OC_Reg field11;
};
struct l_struct_struct_OC_GPR {
   int field0;
  struct l_struct_struct_OC_Reg field1;
   int field2;
  struct l_struct_struct_OC_Reg field3;
   int field4;
  struct l_struct_struct_OC_Reg field5;
   int field6;
  struct l_struct_struct_OC_Reg field7;
   int field8;
  struct l_struct_struct_OC_Reg field9;
   int field10;
  struct l_struct_struct_OC_Reg field11;
   int field12;
  struct l_struct_struct_OC_Reg field13;
   int field14;
  struct l_struct_struct_OC_Reg field15;
   int field16;
  struct l_struct_struct_OC_Reg field17;
   int field18;
  struct l_struct_struct_OC_Reg field19;
   int field20;
  struct l_struct_struct_OC_Reg field21;
   int field22;
  struct l_struct_struct_OC_Reg field23;
   int field24;
  struct l_struct_struct_OC_Reg field25;
   int field26;
  struct l_struct_struct_OC_Reg field27;
   int field28;
  struct l_struct_struct_OC_Reg field29;
   int field30;
  struct l_struct_struct_OC_Reg field31;
   int field32;
  struct l_struct_struct_OC_Reg field33;
};
struct l_struct_struct_OC_anon_OC_3 {
   int field0;
   int field1;
};
struct l_array_8_struct_AC_l_struct_struct_OC_anon_OC_3 {
  struct l_struct_struct_OC_anon_OC_3 array[8];
};
struct l_struct_struct_OC_X87Stack {
  struct l_array_8_struct_AC_l_struct_struct_OC_anon_OC_3 field0;
};
struct l_array_1_ureplace_u64int {
   int array[1];
};
struct l_struct_struct_OC_uint64v1_t {
  struct l_array_1_ureplace_u64int field0;
};
struct l_struct_union_OC_vec64_t {
  struct l_struct_struct_OC_uint64v1_t field0;
};
struct l_struct_struct_OC_anon_OC_4 {
   int field0;
  struct l_struct_union_OC_vec64_t field1;
};
struct l_array_8_struct_AC_l_struct_struct_OC_anon_OC_4 {
  struct l_struct_struct_OC_anon_OC_4 array[8];
};
struct l_struct_struct_OC_MMX {
  struct l_array_8_struct_AC_l_struct_struct_OC_anon_OC_4 field0;
};
struct l_struct_struct_OC_FPUStatusFlags {
   int field0;
   int field1;
   int field2;
   int field3;
   int field4;
   int field5;
   int field6;
   int field7;
   int field8;
   int field9;
   int field10;
   int field11;
   int field12;
   int field13;
   int field14;
   int field15;
   int field16;
   int field17;
   int field18;
   int field19;
  struct l_array_4_ureplace_u8int field20;
};
struct l_struct_union_OC_FPUAbridgedTagWord {
   int field0;
};
struct l_struct_union_OC_FPUControlStatus {
   int field0;
};
struct l_array_10_ureplace_u8int {
   int array[10];
};
struct l_struct_struct_OC_float80_t {
  struct l_array_10_ureplace_u8int field0;
};
struct l_struct_union_OC_anon_OC_11 {
  struct l_struct_struct_OC_float80_t field0;
};
struct l_array_6_ureplace_u8int {
   int array[6];
};
struct l_struct_struct_OC_FPUStackElem {
  struct l_struct_union_OC_anon_OC_11 field0;
  struct l_array_6_ureplace_u8int field1;
};
struct l_array_8_struct_AC_l_struct_struct_OC_FPUStackElem {
  struct l_struct_struct_OC_FPUStackElem array[8];
};
struct l_array_96_ureplace_u8int {
   int array[96];
};
struct l_struct_struct_OC_SegmentShadow {
  struct l_struct_union_OC_anon field0;
   int field1;
   int field2;
};
struct l_struct_struct_OC_SegmentCaches {
  struct l_struct_struct_OC_SegmentShadow field0;
  struct l_struct_struct_OC_SegmentShadow field1;
  struct l_struct_struct_OC_SegmentShadow field2;
  struct l_struct_struct_OC_SegmentShadow field3;
  struct l_struct_struct_OC_SegmentShadow field4;
  struct l_struct_struct_OC_SegmentShadow field5;
};
struct l_struct_struct_OC_State {
  struct l_struct_struct_OC_ArchState field0;
  struct l_array_32_struct_AC_l_struct_union_OC_VectorReg field1;
  struct l_struct_struct_OC_ArithFlags field2;
  struct l_struct_union_OC_anon field3;
  struct l_struct_struct_OC_Segments field4;
  struct l_struct_struct_OC_AddressSpace field5;
  struct l_struct_struct_OC_GPR field6;
  struct l_struct_struct_OC_X87Stack field7;
  struct l_struct_struct_OC_MMX field8;
  struct l_struct_struct_OC_FPUStatusFlags field9;
  struct l_struct_union_OC_anon field10;
  struct l_struct_struct_OC_SegmentCaches field12;
};

/* External Global Variable Declarations */
extern struct l_struct_struct_OC_State* globalState;

/* Function Declarations */
sub_401106_foo(struct l_struct_struct_OC_State* tmp__14,  int tmp__15, void* tmp__16);
void* sub_401106___VERIFIER_nondet_int(struct l_struct_struct_OC_State* tmp__39,  int tmp__40, void* tmp__41);
extern void __VERIFIER_nondet_unsigned() __attribute__ ((__));

extern void __VERIFIER_assume() __attribute__ ((__noreturn__));

/* Global Variable Definitions and Initialization */
int x;
int y;
int STATE_REG_RAX ;


/* LLVM Intrinsic Builtin Function Bodies */
static   int llvm_add_u32( int a,  int b) {
   int r = a + b;
  return r;
}
static   int llvm_lshr_u32( int a,  int b) {
   int r = a >> b;
  return r;
}
static   int llvm_and_u8( int a,  int b) {
   int r = a & b;
  return r;
}
static   int llvm_xor_u8( int a,  int b) {
   int r = a ^ b;
  return r;
}


/* Function Bodies */

 void* main(struct l_struct_struct_OC_State* tmp__1,  int tmp__2, void* tmp__3) {
  struct l_struct_struct_OC_State* tmp__4;
   int* tmp__5;
   int* tmp__6;
   int* tmp__7;
   int* tmp__8;
   int* tmp__9;
   int* tmp__10;
  void* tmp__11;
   int tmp__12;
   int tmp__13;
   int tmp__14;
   int tmp__14__PHI_TEMPORARY;
   int tmp__15;
   int tmp__16;
   int tmp__17;
   int tmp__18;
   int _2e_lcssa3;
   int _2e_lcssa3__PHI_TEMPORARY;
  int _2e_lcssa2;
  int _2e_lcssa2__PHI_TEMPORARY;
   int _2e_lcssa1;
   int _2e_lcssa1__PHI_TEMPORARY;
   int tmp__19;

  tmp__4 = globalState;
  tmp__5 = (&tmp__4->field2.field1);
  tmp__6 = (&tmp__4->field2.field3);
  tmp__7 = (&tmp__4->field2.field7);
  tmp__8 = (&tmp__4->field2.field9);
  tmp__9 = (&tmp__4->field2.field13);
  tmp__10 = (&tmp__4->field2.field5);
  goto block_401119;

  do {     /* Syntactic loop 'block_401119' to make GCC happy */
block_401119:
     STATE_REG_RAX = __VERIFIER_nondet_int();
 
    tmp__11 =  /*tail*/ sub_401106___VERIFIER_nondet_int(/*UNDEF*/((struct l_struct_struct_OC_State*)/*NULL*/0), /*UNDEF*/(0UL), tmp__3);
    tmp__12 = STATE_REG_RAX;
    x = tmp__12;
    y = 1;
    tmp__13 = tmp__12 >> 31;
  

  /* if (((((((tmp__13 == 0u)&1)) & (((~((((tmp__12 == 0u)&1))))&1)))&1))) { */
    if (((((tmp__13 == 0u)&1)) && (tmp__12 != 0u))&1) {

    tmp__14__PHI_TEMPORARY = tmp__12;   /* for PHI node */
    goto block_401135;
  } else {
    _2e_lcssa3__PHI_TEMPORARY = tmp__12;   /* for PHI node */
    _2e_lcssa2__PHI_TEMPORARY = (((tmp__12 == 0u)&1));   /* for PHI node */
    _2e_lcssa1__PHI_TEMPORARY = tmp__13;   /* for PHI node */
    goto block_401163;
  }
 
  do {     /* Syntactic loop 'block_401135' to make GCC happy */
block_401135:
  tmp__14 = tmp__14__PHI_TEMPORARY;
  tmp__15 = tmp__14-1;
  tmp__16 = tmp__14-2;
  tmp__17 = tmp__16>>31;
  tmp__18 = tmp__15>>31;
 
  /* if (((((((tmp__16 != 0u)&1)) & ((((((tmp__17 == 0u)&1)) ^ ((((llvm_add_u32((tmp__17 ^ tmp__18), tmp__18)) == 2u)&1)))&1)))&1))) { */
  if ( ((tmp__16 != 0u)&1) &&(tmp__17 == 0u)) {
    goto block_401159_2e_backedge;
  } else {
    goto block_40114f;
  }

block_40114f:
  y = 0;
  goto block_401159_2e_backedge;

block_401159_2e_backedge:
  /* if (((((((tmp__18 == 0u)&1)) & (((~((((tmp__15 == 0u)&1))))&1)))&1))) { */

    if (((((((tmp__18 == 0u)&1)) && ( tmp__15 != 0u) )&1))) {
    tmp__14__PHI_TEMPORARY = tmp__15;   /* for PHI node */
    goto block_401135;
  } else {
    goto block_401159_2e_block_401163_crit_edge;
  }

  } while (1); /* end of syntactic loop 'block_401135' */
block_401159_2e_block_401163_crit_edge:
   x = tmp__15;
  _2e_lcssa3__PHI_TEMPORARY = tmp__15;   /* for PHI node */
  _2e_lcssa2__PHI_TEMPORARY = (((tmp__15 == 0u)&1));   /* for PHI node */
  _2e_lcssa1__PHI_TEMPORARY = tmp__18;   /* for PHI node */
  goto block_401163;

block_401163:
  _2e_lcssa3 = _2e_lcssa3__PHI_TEMPORARY;
  _2e_lcssa2 = ((_2e_lcssa2__PHI_TEMPORARY)&1);
  _2e_lcssa1 = _2e_lcssa1__PHI_TEMPORARY;
  tmp__19 =  /*tail*/ llvm_OC_ctpop_OC_i32((_2e_lcssa3 & 255));
  *tmp__5 = 0;
  *tmp__6 =  ( ((((int)tmp__19))& 1))^ 1 ;
  *tmp__7 = ((( int)(int)_2e_lcssa2));
  *tmp__8 = ((( int)_2e_lcssa1));
  *tmp__9 = 0;
  *tmp__10 = 0;
  goto block_401119;

  } while (1); /* end of syntactic loop 'block_401119' */
}

\end{lstlisting}

\vfill
\pagebreak
\renewcommand\ruleLabel[1]{[{\sc #1}]}
\section{Bitwise Branching Rules with Labels}
\label{apx:rules}
The rules from Section~\ref{sec:bitwise} are reproduced here, now with labels.
\begin{center}
  \begin{table}
    \footnotesize
\begin{tabular}{rllll}

  \rewrite{e_1=0}{e_1 \texttt{\&} e_2}{0} & \ruleLabel{R-And-0}\\

  \rewrite{(e_1=0 \vee e_1=1) \wedge e_2=1}{e_1 \texttt{\&} e_2}{e_1} & \ruleLabel{R-And-1}\\

  \rewrite{(e_1=0 \vee e_1=1) \wedge (e_2=0 \vee e_2=1)}{e_1 \texttt{\&} e_2}{e_1 \texttt{\&\&} e_2} & \ruleLabel{R-And-LOG}\\
 
  \rewrite{e_1 \geq 0 \wedge e_2 = 1}{e_1 \texttt{\&} e_2}{e_1 \texttt{\%} 2} & \ruleLabel{R-And-LBS}\\
  
  \rewrite{e_2=0}{e_1 \texttt{|} e_2}{e_1} & \ruleLabel{R-Or-0}\\
  
  \rewrite{(e_1=0 \vee e_1=1) \wedge e_2=1}{e_1 \texttt{|} e_2}{1} & \ruleLabel{R-Or-1} \\

  \rewrite{e_2=0}{e_1 \texttt{\^{}} e_2}{e_1} & \ruleLabel{R-Xor-0}\\
 
  \rewrite{e_1=e_2=0 \vee e_1=e_2=1}{e_1 \texttt{\^{}} e_2}{0} & \ruleLabel{R-Xor-Eq}\\
  
  \rewrite{(e_1=1 \wedge e_2=0) \vee (e_1=0 \wedge e_2=1)}{e_1 \texttt{\^{}} e_2}{1} & \ruleLabel{R-Xor-Neq}\\

  \rewrite{e_1 \geq 0 \wedge e_2 = \texttt{CHAR\_BIT * sizeof}(e_1) - 1}{e_1 \texttt{>>} e_2}{0} & \ruleLabel{R-RightShift-Pos}\\
  
  \rewrite{e_1 < 0 \wedge e_2 = \texttt{CHAR\_BIT * sizeof}(e_1) - 1}{e_1 \texttt{>>} e_2}{-1} & \ruleLabel{R-RightShift-Neg}

\end{tabular}

  \end{table}
  
  \begin{table}
\begin{adjustbox}{width=\textwidth}
    
\end{adjustbox}
  \end{table}
\end{center}

\section{Bitwise Branching Algorithm}
\label{apx:alg}
\begin{tabular}{|l|l|}
\hline
\begin{minipage}[t]{2.3in}
\begin{lstlisting}[language=caml, basicstyle=\tt\scriptsize]
type rule_exp = (exp -> exp -> exp) * (exp -> exp -> exp) 
let rec |$T_E$| (e:exp) : exp =
  match e with
  |$\mid$| BinOp(|$\otimes$|,e1,e2) ->
    let e1' = |$T_E$| e1 in
    let e2' = |$T_E$| e2 in
    let rules = Rules.find_exp(|$\otimes$|) in
    fold_left (fun acc (cond,repl) ->
      ITE(cond e1' e2',repl e1' e2',acc)
    )  (BinOp(|$\otimes$|,e1, e2)) rules
  |$\mid$| _ -> e
\end{lstlisting}
\end{minipage}
&
\begin{minipage}[t]{2.3in}
\begin{lstlisting}[language=caml, basicstyle=\tt\scriptsize]
type rule_stmt = (exp -> exp -> exp) * (lhs -> exp -> exp -> stmt)
let |$T_s$| (s:stmt) : stmt =
  match s with
  |$\mid$| Assign(lhs,BinOp(|$\otimes$|,e1,e2)) ->
    let e1' = |$T_E$| e1 in
    let e2' = |$T_E$| e2 in
    let rules = Rules.find_stmt(|$\otimes$|)in
    fold_left (fun acc (cond,repl) ->
      IfElseStmt(cond e1' e2',repl lhs e1' e2', acc)
    ) (Assign(l, BinOp(|$\otimes$|, e1, e2) rules
   |$\mid$| _ -> s
\end{lstlisting}
\end{minipage}\\
\hline
\end{tabular}

\vfill
\pagebreak
\section{Proofs of Bitwise Branching Rules}
\label{apx:rules:proof}


\begin{lstlisting}[language=Python, basicstyle=\scriptsize]
from z3 import *

def prove(r, f):
    s = Solver()
    s.add(Not(f))
    if s.check() == unsat:
        print ("proved rule: " + r)
    else:
        print ("failed to prove rule: " + r)
        print(s.model())

def vec2bool(v):
    if v!=0: return True
    else: return False

def bool2vec(b):
    if (b==True): return BitVecVal(1,1)
    else: return BitVecVal(0,1)

# prove RS-POS rule
rule_1 = "RS-POS: check((n>>31)==0) <==> n>=0"
n_1 = BitVec('n_1', 32)
A_1 = ((n_1>>31)==0)
B_1 = (n_1 >= 0)
constraints_1=And (Implies (A_1, B_1), Implies (B_1, A_1))
prove(rule_1, constraints_1)

# prove RS-NEG rule
# 32 bit 2's complement
rule_2 = "RS-NEG: check((n>>31)==-1) <==> n<0"
n_2 = BitVec('n_2', 32)
A_2 = ((n_2>>31)==-1)
B_2 = (n_2 < 0)

constraints_2=And(Implies(A_2, B_2), Implies(B_2, A_2))
prove(rule_2, constraints_2)

# prove AND-1 rule
rule-3="AND-1: check(n&1) <==> n"
n_3 = BitVec('n_3', 32)
A_3 = ((n_3&1) == n_3)
B_3 = Or(n_3==0, n_3==1)
constraints_3= And (Implies(A_3, B_3), Implies (B_3, A_3))
prove(rule_3, constraints_3)

# prove AND-0 rule
rule_4 = "AND-0: check(n&0) <==> 0"
n_4 = BitVec('n_4', 32)
A_4 = ((n_4 & 0) == 0)
prove(rule_4, A_4)
# print ("AND-0 rule: check(n&0)==0," + str(s4.check()))

# prove XOR-0 rule
rule_6="XOR-0: check(n^0) <==> n"
n_6 = BitVec('n_6', 32)
A_6 = ((n_6^0) == n_6)
prove(rule_6, A_6)

# prove XOR-EQ rule
rule_6_1 = "XOR-EQ: check(n=1 or 0, e = 1 or 0, n==e ) ==> 0"
n_6_1, e_6_1 = BitVecs('n_6_1 e_6_1', 32)
A_6_1 = ((n_6_1^e_6_1) == 0)
B_6_1 = And(Or(n_6_1 ==1, n_6_1 == 0), Or(e_6_1 ==1, e_6_1 == 0), n_6_1 == e_6_1)
prove(rule_6_1, Implies(B_6_1, A_6_1))

# prove XOR-NEQ rule
rule_6_2 = "XOR-NEQ: check(n=1 or 0, e = 1 or 0, n!=e ) ==> 1"
n_6_2, e_6_2 = BitVecs('n_6_2 e_6_2', 32)
A_6_2 = ((n_6_2^e_6_2) == 1)
B_6_2 = And(Or(n_6_2 == 1, n_6_2 == 0), Or(e_6_2 == 1, e_6_2 == 0), n_6_2 != e_6_2)
prove(rule_6_2, Implies(B_6_2, A_6_2))

rule_7="NOT-SWITCH: check(~(n==b) <==> n!=b)"
n_7, b = BitVecs('n_7 b', 32)
A_7 = (Not(n_7==b))
B_7 = (n_7!=b)
constraints_7 = And(Implies(A_7, B_7))
prove(rule_7, constraints_7)

# prove OR-1 rule
rule_8="OR-1: n$|$1 <==> 1 "
n_8 = BitVec('n_8', 32)
A_8 = ((n_8$|$1)==1)
B_8 = Or (n_8 ==1, n_8==0)
prove(rule_8, And (Implies(B_8, A_8), Implies(A_8, B_8)))

# prove OR-0 rule
rule_9="OR-0: n$|$0 <==> n "
n_9 = BitVec('n_9', 32)
A_9 = ((n_9$|$0)==n_9)
prove(rule_9, A_9)


# prove BOOL-TRUE rule
rule_10="BOOL-TRUE: n!=0 ==> True "
n_10 = BitVec('n_10', 32)
A_10 = Implies(n_10!=0,True)
prove(rule_10, A_10)
  
# General rules &, operators are mutual exclusive negative
rule_and_xneg = "r=a&b exclusive negative: r = a&b, a>0, b<0 ==> r<=a, r>0"
and_a3, and_b3, and_r3= BitVecs('and_a3 and_b3 and_r3', 32)
constr_and_xneg = Implies(And(and_a3 >=0 , and_b3<0, and_r3==(and_a3&and_b3)), And(and_r3 <=and_a3, and_r3>=0))
prove(rule_and_xneg, constr_and_xneg)

# General rules &, both operators are non-negative
rule_and4 = "r<=a&b less than: r <= a&b, a>=0, b>=0 ==> r<=a, r<=b"
and_a4, and_b4, and_r4= BitVecs('and_a4 and_b4 and_r4', 32)
constr_and4 = Implies(And(and_a4 >= 0, and_b4 >= 0, and_r4 <= (and_a4 & and_b4)), And(and_r4 < and_a4, and_r4 < and_b4))
prove(rule_and4, constr_and4)

# General rules &, both operators are negative
rule_and_neg5 = "r=a&b both negative: r < a&b, a<0, b<0 ==> r<=a, r<=b, r<0"
and_a5, and_b5, and_r5= BitVecs('and_a5 and_b5 and_r5', 32)
constr_and_neg5 = Implies(And(and_a5 < 0, and_b5<0, and_r5<(and_a5&and_b5)), And(and_r5 < and_a5, and_r5 < and_b5, and_r5<0))
prove(rule_and_neg5, constr_and_neg5)

# And logic rules &, both operators are one bit size
rule_and_neg6 = "a&b both one bit: a&b -> a&&b"
and_a6, and_b6 = BitVecs('and_a6 and_b6', 1)
constr_and_neg6 = Implies((and_a6&and_b6)==1, And(and_a6 ==1, and_b6 ==1))
prove(rule_and_neg6, constr_and_neg6)

# General rules or, both operators are non-negative
rule_or = "r=a$|$b genral: r >= a$|$b, a>=0, b>=0 ==> r>=a, r>=b"
or_a1, or_b1, or_r1= BitVecs('or_a1 or_b1 or_r1', 32)
constr_or = Implies(And(or_a1 >= 0, or_b1>=0, or_r1>=(or_a1$|$or_b1)), And(or_r1 >= or_a1, or_r1 >= or_b1))
prove(rule_or, constr_or)

# General rules or, both operators are negative(two's complement) 
rule_or_neg = "r=a$|$b genral both negative: r = a$|$b, a<0, b<0 ==> r>=a, r>=b, r<0"
or_a2, or_b2, or_r2= BitVecs('or_a2 or_b2 or_r2', 32)
constr_or_neg = Implies(And(or_a2 < 0, or_b2 < 0, or_r2==(or_a2$|$or_b2)), And(or_r2 >= or_a2, or_r2 >= or_b2, or_r2<0))
prove(rule_or_neg, constr_or_neg)

# General rules or, both operators are mutual exclusicve negative(two's complement) 
rule_or_xneg = "r=a$|$b genral mutual exclusive negative: r = a$|$b, a<0, b>=0 ==> r>=a, r<0"
or_a3, or_b3, or_r3= BitVecs('or_a3 or_b3 or_r3', 32)
constr_or_xneg = Implies(And(or_a3 >= 0, or_b3 < 0, or_r3==(or_a3$|$or_b3)), And(or_r3 >= or_b3, or_r3<0))
prove(rule_or_xneg, constr_or_xneg)

# Or logic rule, both operators are one bit
rule_or_log = "a$|$b == 0 to logic: (a$|$b) == 0 ==> a==0 && b==0 ==> r>=a, r>=b"
or_a4, or_b4 = BitVecs('or_a4 or_b4', 1)
constr_or_log = Implies((or_a4$|$or_b4 == 0), And(or_a4==0,or_b4 == 0))
prove(rule_or_log, constr_or_log)

# General rules xor, both operators are non-negative
rule_xor = "r>=a^b genral, both non-negative: r >= a^b, a>=0, b>=0 ==> r>=0"
xor_a, xor_b, xor_r= BitVecs('xor_a xor_b xor_r', 32)
constr_xor = Implies(And(xor_a >= 0, xor_b>=0, xor_r>=(xor_a^xor_b)), And(xor_r >= 0))
prove(rule_xor, constr_xor)

# General rules xor, both operators are negative(two's complement) 
rule_xor_neg = "r>=a^b genral both negative: r >= a^b, a<0, b<0 ==> r<0"
xor_a1, xor_b1, xor_r1= BitVecs('xor_a1 xor_b1 xor_r1', 32)
constr_xor_neg = Implies(And(xor_a1 < 0, xor_b1 < 0, xor_r1>=(xor_a1^xor_b1)), And(xor_r1>=0))
prove(rule_xor_neg, constr_xor_neg)

# General rules xor, operators are mutual exlusive non-negative(two's complement) 
rule_xor_xneg = "r<=a^b genral mutual exlusive negative: r <= a^b, a<0, b>=0 ==> r<0"
xor_a2, xor_b2, xor_r2= BitVecs('xor_a2 xor_b2 xor_r2', 32)
constr_xor_xneg = Implies(And(xor_a2 < 0, xor_b2 >= 0, xor_r2<=(xor_a2^xor_b2)), And(xor_r2<=0))
prove(rule_xor_xneg, constr_xor_xneg)

# xor to logic rule, operators are one bit 
rule_xor_log = "a^b one  bit ==> (a==0 && b == 1 || a==1 && b==0)"
xor_a3, xor_b3 = BitVecs('xor_a3 xor_b3', 1)
constr_xor_log = Implies(1==(xor_a3^xor_b3), Or(And(xor_a3==0, xor_b3==1), And(xor_a3==1,xor_b3 ==0)))
prove(rule_xor_log, constr_xor_log)

# General rules complement
rule_com = "CPL-POS & CPL-NEG, r= ~a genral"
com_a, com_r= BitVecs('com_a com_r', 32)
constr_com1 = Implies(And(com_a >= 0, com_r==(~com_a)), com_r <0)
constr_com2 = Implies(And(com_a < 0, com_r==(~com_a)), com_r>=0)
constr_com = And(constr_com1, constr_com2)
prove(rule_com, constr_com)
 
\end{lstlisting}

\vfill
\pagebreak
\section{Detailed Results for Termination and LTL}
\label{apx:details:ltl:term:source}
Table \ref{tab:ltlbithacks} shows the details running results of 26 bithacks benchmarks on LTL verification task, Table \ref{tab:ltlcyrules} and Table \ref{tab:termcyrules} show the details of LTL and Termination running results on our hand-crafted benchmarks, table \ref{tab:termaprove} shows the details of termination running results on 18 bitiwse benchmarks from \aprove, they are corresponding to section \ref{sec:termination}.

\begin{center}
    \begin{table}
        \centering
        \caption{\label{tab:ltlbithacks} Details for LTL Bithack benchmarks.}
        \resizebox{\textwidth}{!}{\begin{tabular}{lllcrcr}
    \toprule
                                            &                                &          & \multicolumn{2}{c}{\ultimate} & \multicolumn{2}{c}{\Tool}                    \\
    Benchmark                               & Property                       & Expected & Time                          & Result                    & Time    & Result \\
    \cmidrule(r){1-1}\cmidrule(lr){2-3}\cmidrule(lr){4-5}\cmidrule(l){6-7}
    counting-bits-BK1\_false.c              & $\square(\lozenge y>=0)      $ & \NOK     & 8.80s                         & \UNK                      & 10.75s  & \NOK   \\
    consecutive-zero-bits-trailing\_false.c & $\lozenge y=1                $ & \NOK     & 5.91s                         & \UNK                      & 7.24s   & \UNK   \\
    counting-bits-BK\_false.c               & $\square(\lozenge y<=1)      $ & \NOK     & 6.93s                         & \UNK                      & 8.15s   & \NOK   \\
    display-bit1\_false.c                   & $\lozenge y>1                $ & \NOK     & 24.55s                        & \UNK                      & 59.18s  & \NOK   \\
    parity\_false.c                         & $\lozenge y>=1               $ & \NOK     & 23.81s                        & \UNK                      & 9.09s   & \NOK   \\
    display-bit\_false.c                    & $\lozenge y<0                $ & \NOK     & 27.20s                        & \UNK                      & 64.63s  & \NOK   \\
    counting-bits-set\_false.c              & $\lozenge y>=1               $ & \NOK     & 8.88s                         & \UNK                      & 7.77s   & \NOK   \\
    reverse-bits1\_false.c                  & $\lozenge n<0                $ & \NOK     & 8.11s                         & \UNK                      & 10.21s  & \NOK   \\
    logbase2.c                              & $\lozenge y>=1               $ & \OK      & 7.60s                         & \UNK                      & 7.49s   & \UNK   \\
    base64\_ltl.c                           & $\square(\lozenge start = 1) $ & \OK      & 124.37s                       & \OOM                      & 602.61s & \OOM   \\
    modulus-division.c                      & $\lozenge y>1                $ & \OK      & 6.67s                         & \UNK                      & 17.25s  & \UNK   \\
    consecutive-zero-bits-trailing.c        & $\lozenge y>=1               $ & \OK      & 6.07s                         & \UNK                      & 7.51s   & \OK    \\
    interleave-bits.c                       & $\lozenge y>=1               $ & \OK      & 11.40s                        & \UNK                      & 300.66s & \UNK   \\
    logbase2-N-bit1.c                       & $\lozenge y>=1               $ & \OK      & 12.39s                        & \UNK                      & 547.43s & \OOM   \\
    reverse-N-bit.c                         & $\lozenge n>=1               $ & \OK      & 6.91s                         & \OK                       & 13.09s  & \OK    \\
    counting-bits-set.c                     & $\lozenge y>=1               $ & \OK      & 9.23s                         & \UNK                      & 8.52s   & \OK    \\
    consecutive-zero-bits.c                 & $\lozenge y>=1               $ & \OK      & 5.29s                         & \OK                       & 5.85s   & \OK    \\
    counting-bits-BK.c                      & $\square(\lozenge y<=1)      $ & \OK      & 6.63s                         & \UNK                      & 8.46s   & \OK    \\
    dropbf\_ltl.c                           & $\square(A!=1 \vee RELEASE=0)$ & \OK      & 8.42s                         & \OK                       & 13.47s  & \OK    \\
    reverse-bits.c                          & $\lozenge y>=1               $ & \OK      & 7.21s                         & \UNK                      & 10.81s  & \UNK   \\
    counting-bits-lookup.c                  & $\lozenge y>=1               $ & \OK      & 900.41s                       & \TOUT                     & 900.41s & \TOUT  \\
    display-bit.c                           & $\lozenge y>=32              $ & \OK      & 24.99s                        & \UNK                      & 705.44s & \OOM   \\
    counting-bits-BK1.c                     & $\square(\lozenge y>=0)      $ & \OK      & 7.51s                         & \UNK                      & 8.59s   & \OK    \\
    display-bit1.c                          & $\lozenge y>=1               $ & \OK      & 20.46s                        & \UNK                      & 7.29s   & \OK    \\
    reverse-bits1.c                         & $\lozenge n<0                $ & \OK      & 6.33s                         & \UNK                      & 12.55s  & \OK    \\
    parity.c                                & $\lozenge y>=1               $ & \OK      & 19.67s                        & \UNK                      & 6.93s   & \OK    \\
    \bottomrule
\end{tabular}}
    \end{table}
    
    \begin{table}
        \centering
        \caption{\label{tab:ltlcyrules} Details for \cyrulesLTL.}
        \begin{tabular}{llcrcrc}
    \toprule
                         &                           &          & \multicolumn{2}{c}{\ultimate} & \multicolumn{2}{c}{\Tool}                    \\
    Benchmark            & Property                  & Expected & Time                          & Result                    & Time    & Result \\
    \cmidrule(r){1-1}\cmidrule(lr){2-3}\cmidrule(lr){4-5}\cmidrule(l){6-7}
    and\_guard2\_false.c & $\lozenge z \geq 100    $ & \NOK     & 7.54s                         & \UNK                      & 7.65s   & \NOK   \\
    xor\_stem1\_false.c  & $\lozenge n<0           $ & \NOK     & 5.78s                         & \UNK                      & 6.41s   & \NOK   \\
    and\_guard\_false.c  & $\lozenge y>0           $ & \NOK     & 7.11s                         & \UNK                      & 5.52s   & \NOK   \\
    and\_stem\_false.c   & $\lozenge n<0           $ & \NOK     & 6.34s                         & \UNK                      & 5.86s   & \NOK   \\
    xor\_stem\_false.c   & $\lozenge n<0           $ & \NOK     & 5.67s                         & \UNK                      & 6.33s   & \NOK   \\
    xor\_guard\_false.c  & $\lozenge n>0           $ & \NOK     & 6.40s                         & \UNK                      & 6.24s   & \NOK   \\
    or\_loop1\_false.c   & $\square(\lozenge n<0)  $ & \NOK     & 6.98s                         & \UNK                      & 6.76s   & \NOK   \\
    and\_loop\_false.c   & $\lozenge y \geq 1      $ & \NOK     & 5.32s                         & \UNK                      & 9.28s   & \NOK   \\
    and\_stem1\_false.c  & $\lozenge n<0           $ & \NOK     & 7.63s                         & \UNK                      & 5.75s   & \NOK   \\
    xor\_loop\_false.c   & $\lozenge n<0           $ & \NOK     & 6.18s                         & \UNK                      & 8.01s   & \NOK   \\
    or\_guard\_false.c   & $\lozenge n<0           $ & \NOK     & 5.82s                         & \UNK                      & 5.96s   & \NOK   \\
    and\_guard1\_false.c & $\lozenge n>0           $ & \NOK     & 7.82s                         & \UNK                      & 6.04s   & \NOK   \\
    com\_loop\_false.c   & $\lozenge y<0           $ & \NOK     & 5.70s                         & \UNK                      & 6.81s   & \NOK   \\
    or\_stem\_false.c    & $\lozenge n<0           $ & \NOK     & 7.11s                         & \UNK                      & 5.81s   & \NOK   \\
    and\_guard4\_false.c & $\lozenge n>0           $ & \NOK     & 8.62s                         & \UNK                      & 5.55s   & \NOK   \\
    and\_stem2\_false.c  & $\lozenge n>0           $ & \NOK     & 7.69s                         & \UNK                      & 6.35s   & \NOK   \\
    or\_loop2\_false.c   & $\lozenge n>0           $ & \NOK     & 6.47s                         & \UNK                      & 8.89s   & \NOK   \\
    and\_loop1\_false.c  & $\lozenge z<0           $ & \NOK     & 9.22s                         & \UNK                      & 7.60s   & \NOK   \\
    com\_stem\_false.c   & $\square(\lozenge y<0)  $ & \NOK     & 8.89s                         & \UNK                      & 9.52s   & \NOK   \\
    or\_loop\_false.c    & $\square(\lozenge n<0)  $ & \NOK     & 7.60s                         & \UNK                      & 7.82s   & \NOK   \\
    and\_stem2.c         & $\lozenge n>0           $ & \OK      & 5.81s                         & \UNK                      & 5.36s   & \OK    \\
    xor\_stem1.c         & $\lozenge n<0           $ & \OK      & 7.05s                         & \UNK                      & 5.76s   & \OK    \\
    xor\_guard.c         & $\lozenge n<0           $ & \OK      & 6.41s                         & \UNK                      & 5.67s   & \OK    \\
    and\_guard4.c        & $\lozenge n>0           $ & \OK      & 10.25s                        & \UNK                      & 6.00s   & \OK    \\
    or\_stem.c           & $\lozenge n<0           $ & \OK      & 5.38s                         & \UNK                      & 7.38s   & \OK    \\
    com\_stem.c          & $\square(\lozenge y = 1)$ & \OK      & 6.50s                         & \UNK                      & 5.59s   & \OK    \\
    xor\_stem.c          & $\lozenge n<0           $ & \OK      & 5.40s                         & \UNK                      & 5.41s   & \OK    \\
    xor\_loop.c          & $\lozenge n<0           $ & \OK      & 6.14s                         & \UNK                      & 7.97s   & \OK    \\
    and\_stem1.c         & $\lozenge n<0           $ & \OK      & 7.05s                         & \UNK                      & 5.32s   & \OK    \\
    and\_guard.c         & $\lozenge n<0           $ & \OK      & 6.43s                         & \UNK                      & 901.21s & \TOUT  \\
    and\_loop1.c         & $\lozenge z<0           $ & \OK      & 7.63s                         & \UNK                      & 5.46s   & \OK    \\
    com\_loop.c          & $\lozenge y<0           $ & \OK      & 6.44s                         & \UNK                      & 5.39s   & \OK    \\
    or\_loop.c           & $\square(\lozenge n<0)  $ & \OK      & 7.55s                         & \UNK                      & 10.05s  & \OK    \\
    and\_guard2.c        & $\lozenge z \geq 100    $ & \OK      & 9.58s                         & \UNK                      & 17.04s  & \OK    \\
    or\_loop2.c          & $\lozenge n<0           $ & \OK      & 6.17s                         & \UNK                      & 6.33s   & \OK    \\
    and\_stem.c          & $\lozenge n<0           $ & \OK      & 5.49s                         & \UNK                      & 5.57s   & \OK    \\
    or\_0s\_int.c        & $\lozenge p =1          $ & \OK      & 67.11s                        & \UNK                      & 9.38s   & \OK    \\
    and\_loop.c          & $\lozenge y \geq 1      $ & \OK      & 6.15s                         & \UNK                      & 5.92s   & \OK    \\
    and\_guard1.c        & $\lozenge n>0           $ & \OK      & 7.81s                         & \UNK                      & 8.04s   & \OK    \\
    or\_loop1.c          & $\square(\lozenge n<0)  $ & \OK      & 8.44s                         & \UNK                      & 6.58s   & \OK    \\
    or\_loop3.c          & $\square(\lozenge n<0)  $ & \OK      & 7.47s                         & \UNK                      & 6.11s   & \OK    \\
    or\_guard.c          & $\lozenge n<0           $ & \OK      & 5.92s                         & \UNK                      & 4.52s   & \OK    \\
    \bottomrule
\end{tabular}
    \end{table}

    \begin{table}
        \centering
        \caption{\label{tab:termcyrules} Details for \cyrulesTerm.}
        \resizebox{\textwidth}{!}{\begin{tabular}{lcrcrcrcrcrcrc}
    \toprule
                   &          & \multicolumn{2}{c}{\aprove} & \multicolumn{2}{c}{\cpachecker} & \multicolumn{2}{c}{\kittel} & \multicolumn{2}{c}{\twols} & \multicolumn{2}{c}{\ultimate} & \multicolumn{2}{c}{\Tool}                                                        \\
    Benchmark      & Expected & Time                        & Result                          & Time                        & Result                     & Time                          & Result                    & Time    & Result & Time   & Result & Time   & Result \\
    \cmidrule(r){1-1}\cmidrule(lr){2-2}\cmidrule(r){3-4}\cmidrule(lr){5-6}\cmidrule(lr){7-8}\cmidrule(lr){9-10}\cmidrule(lr){11-12}\cmidrule(l){13-14}
    xor-01-false.c & \NOK     & 18.41s                      & \NOK                            & 5.17s                       & \NOK                       & 900.24s                       & \TOUT                     & 900.53s & \TOUT  & 5.54s  & \UNK   & 6.39s  & \NOK   \\
    and-04-false.c & \NOK     & 185.04s                     & \NOK                            & 3.68s                       & \NOK                       & 900.24s                       & \TOUT                     & 0.15s   & \NOK   & 5.33s  & \UNK   & 7.78s  & \NOK   \\
    not-04-false.c & \NOK     & 4.14s                       & \UNK                            & 3.78s                       & \NOK                       & 0.08s                         & \FAIL                     & 0.14s   & \NOK   & 6.84s  & \UNK   & 6.78s  & \NOK   \\
    not-05-false.c & \NOK     & 3.03s                       & \UNK                            & 4.64s                       & \NOK                       & 0.10s                         & \FAIL                     & 0.13s   & \NOK   & 5.50s  & \UNK   & 8.66s  & \NOK   \\
    and-05-false.c & \NOK     & 10.11s                      & \UNK                            & 3.79s                       & \NOK                       & 900.32s                       & \TOUT                     & 0.16s   & \NOK   & 5.60s  & \UNK   & 8.17s  & \NOK   \\
    and-01-false.c & \NOK     & 4.83s                       & \UNK                            & 5.21s                       & \UNK                       & 900.24s                       & \TOUT                     & 900.46s & \TOUT  & 6.27s  & \UNK   & 6.35s  & \NOK   \\
    or-02-false.c  & \NOK     & 13.66s                      & \NOK                            & 3.93s                       & \NOK                       & 900.28s                       & \TOUT                     & 900.39s & \TOUT  & 7.40s  & \UNK   & 8.28s  & \NOK   \\
    not-03-false.c & \NOK     & 2.95s                       & \UNK                            & 3.86s                       & \NOK                       & 0.10s                         & \FAIL                     & 0.17s   & \NOK   & 5.34s  & \UNK   & 7.25s  & \NOK   \\
    and-03-false.c & \NOK     & 4.65s                       & \NOK                            & 3.67s                       & \NOK                       & 900.24s                       & \TOUT                     & 0.14s   & \NOK   & 5.29s  & \UNK   & 5.63s  & \NOK   \\
    not-02-false.c & \NOK     & 1.90s                       & \OK                             & 3.84s                       & \NOK                       & 0.05s                         & \FAIL                     & 0.15s   & \NOK   & 5.26s  & \UNK   & 4.86s  & \NOK   \\
    or-05-false.c  & \NOK     & 6.29s                       & \NOK                            & 4.18s                       & \UNK                       & 900.35s                       & \TOUT                     & 0.16s   & \NOK   & 6.07s  & \UNK   & 9.34s  & \NOK   \\
    or-01-false.c  & \NOK     & 5.76s                       & \NOK                            & 3.73s                       & \NOK                       & 900.31s                       & \TOUT                     & 900.69s & \TOUT  & 7.64s  & \UNK   & 10.22s & \NOK   \\
    and-02-false.c & \NOK     & 6.30s                       & \UNK                            & 5.78s                       & \UNK                       & 900.24s                       & \TOUT                     & 900.39s & \TOUT  & 6.72s  & \UNK   & 5.89s  & \NOK   \\
    not-01.c       & \OK      & 902.06s                     & \TOUT                           & 4.69s                       & \UNK                       & 0.06s                         & \OK                       & 0.39s   & \OK    & 6.54s  & \UNK   & 5.52s  & \OK    \\
    and-02.c       & \OK      & 901.23s                     & \TOUT                           & 8.00s                       & \UNK                       & 900.25s                       & \TOUT                     & 900.38s & \TOUT  & 6.05s  & \UNK   & 11.63s & \OK    \\
    or-01.c        & \OK      & 23.72s                      & \OK                             & 6.50s                       & \UNK                       & 900.29s                       & \TOUT                     & 900.41s & \TOUT  & 7.85s  & \UNK   & 13.10s & \OK    \\
    xor-01.c       & \OK      & 4.14s                       & \OK                             & 4.77s                       & \OK                        & 0.05s                         & \OK                       & 0.31s   & \OK    & 7.16s  & \UNK   & 8.07s  & \OK    \\
    and-03.c       & \OK      & 12.83s                      & \UNK                            & 3.80s                       & \NOK                       & 900.25s                       & \TOUT                     & 900.44s & \TOUT  & 6.36s  & \UNK   & 6.98s  & \OK    \\
    not-02.c       & \OK      & 3.11s                       & \UNK                            & 4.03s                       & \NOK                       & 0.31s                         & \FAIL                     & 0.18s   & \NOK   & 5.09s  & \UNK   & 9.23s  & \OK    \\
    and-01.c       & \OK      & 5.74s                       & \UNK                            & 6.05s                       & \UNK                       & 900.25s                       & \TOUT                     & 900.46s & \TOUT  & 6.04s  & \UNK   & 8.14s  & \OK    \\
    or-02.c        & \OK      & 10.16s                      & \OK                             & 4.14s                       & \UNK                       & 900.29s                       & \TOUT                     & 900.46s & \TOUT  & 7.42s  & \UNK   & 6.56s  & \OK    \\
    or-03.c        & \OK      & 5.27s                       & \OK                             & 4.18s                       & \UNK                       & 0.04s                         & \OK                       & 0.24s   & \OK    & 5.83s  & \OK    & 6.94s  & \OK    \\
    not-03.c       & \OK      & 3.95s                       & \UNK                            & 5.32s                       & \NOK                       & 0.06s                         & \OK                       & 0.27s   & \OK    & 5.55s  & \UNK   & 6.30s  & \OK    \\
    and-04.c       & \OK      & 901.56s                     & \TOUT                           & 3.61s                       & \NOK                       & 900.17s                       & \TOUT                     & 0.21s   & \OK    & 5.50s  & \UNK   & 8.56s  & \OK    \\
    or-06.c        & \OK      & 7.48s                       & \OK                             & 4.22s                       & \UNK                       & 0.04s                         & \OK                       & 0.31s   & \OK    & 7.11s  & \OK    & 6.22s  & \OK    \\
    and-05.c       & \OK      & 4.40s                       & \UNK                            & 3.61s                       & \NOK                       & 900.37s                       & \TOUT                     & 0.19s   & \OK    & 7.04s  & \UNK   & 11.40s & \OK    \\
    not-04.c       & \OK      & 2.39s                       & \UNK                            & 4.70s                       & \UNK                       & 0.10s                         & \OK                       & 900.42s & \TOUT  & 11.49s & \UNK   & 12.85s & \OK    \\
    or-04.c        & \OK      & 17.17s                      & \NOK                            & 4.18s                       & \NOK                       & 900.30s                       & \TOUT                     & 0.15s   & \NOK   & 8.13s  & \UNK   & 5.90s  & \OK    \\
    or-05.c        & \OK      & 5.12s                       & \NOK                            & 6.44s                       & \UNK                       & 900.37s                       & \TOUT                     & 0.16s   & \NOK   & 7.22s  & \UNK   & 8.02s  & \OK    \\
    and-06.c       & \OK      & 10.12s                      & \UNK                            & 5.78s                       & \UNK                       & 900.25s                       & \TOUT                     & 900.56s & \TOUT  & 7.24s  & \UNK   & 9.42s  & \OK    \\
    not-05.c       & \OK      & 3.21s                       & \UNK                            & 4.30s                       & \NOK                       & 0.06s                         & \OK                       & 0.24s   & \OK    & 5.55s  & \UNK   & 5.88s  & \OK    \\
    \bottomrule
\end{tabular}
}
    \end{table}

    \begin{table}
        \centering
        \caption{\label{tab:termaprove} Details for \aprove{} termination benchmarks.}
        \resizebox{\textwidth}{!}{\begin{tabular}{lcrcrcrcrcrcrc}
    \toprule
                                                &          & \multicolumn{2}{c}{\aprove} & \multicolumn{2}{c}{\cpachecker} & \multicolumn{2}{c}{\kittel} & \multicolumn{2}{c}{\twols} & \multicolumn{2}{c}{\ultimate} & \multicolumn{2}{c}{\Tool}                                                         \\
    Benchmark                                   & Expected & Time                        & Result                          & Time                        & Result                     & Time                          & Result                          & Time    & Result & Time   & Result & Time    & Result \\
    \cmidrule(r){1-1}\cmidrule(lr){2-2}\cmidrule(r){3-4}\cmidrule(lr){5-6}\cmidrule(lr){7-8}\cmidrule(lr){9-10}\cmidrule(lr){11-12}\cmidrule(l){13-14}
    signed/wdk-signed-overflow/eeprom2.c        & \OK      & 2.03s                       & \UNK                            & 5.68s                       & \NOK                       & 0.08s                         & \FAIL                           & 0.17s   & \OK    & 22.91s & \UNK   & 130.32s & \OOM   \\
    signed/wdk-signed-overflow/common.c         & \OK      & 2.12s                       & \UNK                            & 5.49s                       & \OK                        & 0.04s                         & \OK                             & 0.25s   & \OK    & 24.38s & \UNK   & 900.44s & \TOUT  \\
    unsigned/juggernaut-paper/a.c               & \OK      & 2.37s                       & \UNK                            & 3.76s                       & \NOK                       & 900.25s                       & \TOUT                           & 0.15s   & \OK    & 5.50s  & \UNK   & 9.72s   & \NOK   \\
    unsigned/pointer/p03.c                      & \OK      & 2.05s                       & \UNK                            & 4.24s                       & \UNK                       & 0.04s                         & \FAIL                           & 0.13s   & \UNK   & 7.16s  & \UNK   & 5.60s   & \UNK   \\
    unsigned/wdk-no-signed-overflow/gsm6102.c   & \OK      & 900.26s                     & \TOUT                           & 3.84s                       & \NOK                       & 900.25s                       & \TOUT                           & 900.38s & \TOUT  & 6.92s  & \OK    & 5.96s   & \OK    \\
    unsigned/wdk-no-signed-overflow/hw\_ccmp.c  & \OK      & 4.30s                       & \UNK                            & 3.59s                       & \FAIL                      & 0.04s                         & \OK                             & 0.13s   & \OK    & 22.03s & \UNK   & 14.77s  & \UNK   \\
    unsigned/wdk-no-signed-overflow/comm.c      & \OK      & 3.88s                       & \OK                             & 5.13s                       & \OK                        & 0.06s                         & \FAIL                           & 0.20s   & \OK    & 19.42s & \UNK   & 22.91s  & \UNK   \\
    unsigned/wdk-no-signed-overflow/gsm6103.c   & \OK      & 900.26s                     & \TOUT                           & 4.76s                       & \NOK                       & 900.25s                       & \TOUT                           & 900.37s & \TOUT  & 6.62s  & \OK    & 8.73s   & \OK    \\
    unsigned/wdk-no-signed-overflow/miniport.c  & \OK      & 900.26s                     & \TOUT                           & 5.90s                       & \UNK                       & 900.25s                       & \TOUT                           & 0.25s   & \OK    & 7.12s  & \NOK   & 8.30s   & \NOK   \\
    unsigned/wdk-no-signed-overflow/namesup2.c  & \OK      & 900.26s                     & \TOUT                           & 3.76s                       & \NOK                       & 0.03s                         & \FAIL                           & 0.59s   & \OK    & 7.67s  & \NOK   & 4.94s   & \NOK   \\
    unsigned/wdk-no-signed-overflow/eeprom.c    & \OK      & 6.71s                       & \UNK                            & 9.65s                       & \OK                        & 900.25s                       & \TOUT                           & 0.13s   & \OK    & 16.79s & \UNK   & 16.54s  & \UNK   \\
    unsigned/wdk-no-signed-overflow/intrface.c  & \OK      & 900.25s                     & \TOUT                           & 3.77s                       & \NOK                       & 900.15s                       & \TOUT                           & 0.17s   & \OK    & 6.61s  & \UNK   & 5.23s   & \NOK   \\
    unsigned/wdk-no-signed-overflow/init.c      & \OK      & 5.56s                       & \UNK                            & 3.69s                       & \NOK                       & 900.14s                       & \TOUT                           & 0.10s   & \OK    & 5.25s  & \UNK   & 13.11s  & \NOK   \\
    unsigned/wdk-no-signed-overflow/namesup.c   & \OK      & 2.39s                       & \UNK                            & 2.01s                       & \FAIL                      & 0.05s                         & \OK                             & 185.56s & \OOM   & 28.44s & \UNK   & 27.67s  & \UNK   \\
    unsigned/wdk-no-signed-overflow/image.c     & \OK      & 900.26s                     & \TOUT                           & 5.06s                       & \NOK                       & 796.47s                       & \TOUT                           & 1.24s   & \OK    & 14.35s & \UNK   & 18.52s  & \UNK   \\
    unsigned/wdk-no-signed-overflow/mp\_util.c  & \OK      & 3.82s                       & \UNK                            & 3.80s                       & \UNK                       & 900.27s                       & \TOUT                           & 0.95s   & \OK    & 30.24s & \UNK   & 26.87s  & \UNK   \\
    unsigned/wdk-no-signed-overflow/allocsup1.c & \OK      & 900.26s                     & \TOUT                           & 3.78s                       & \NOK                       & 0.07s                         & \FAIL                           & 0.17s   & \OK    & 6.73s  & \UNK   & 5.62s   & \UNK   \\
    unsigned/juggernaut/loop6.c                 & \OK      & 2.33s                       & \UNK                            & 3.91s                       & \NOK                       & 900.25s                       & \TOUT                           & 0.15s   & \OK    & 5.76s  & \UNK   & 7.55s   & \NOK   \\
    \bottomrule
\end{tabular}
}
    \end{table}
    
\end{center}

\section{Bug in GCC}
\label{apx:gccbug}

\begin{example}[A reported bug in {\tt gcc}~\cite{gcc-bug3:online}]
\begin{lstlisting}[language=C]
static uint64_t div(uint64_t ui1, uint64_t ui2){
  return (ui2 == 0) ? ui1 : (ui1 / ui2); }
static int8_t mod(int8_t si1, int8_t si2){
  return (si2==0) $\mid\mid$ ((si1==-128) && (si2==-1)) ?  si1 : (si1 % si2); }
static int32_t g_5=0, g_11=0;
int main (){
  uint64_t l_7 = 0x509CB0BEFCDF11BBLL;
  g_11 ^= l_7 && ((div((mod(g_5, 0)), -1L)) != 1L);
  if (!g_11) return __VERIFIER_error();
  return 0;
}
\end{lstlisting}
\end{example}
From the source semantics of this example,
variable {\tt g\_11}
would have value 1 at Line 9 and {\tt \_\_VERIFIER\_error} is 
un-reachable. However, a defect~\cite{gcc-bug2:online} 
in {\tt gcc} folds {\tt ((div((mod(g\_5, 0)), -1L)) != 1L)} to 0. This makes 
{\tt g\_11} become 0 at Line 9, introducing a new error behavior.

\section{\Tool\ translations through an example lifted binary}

\label{apx:lifting-challenges}

\begin{figure}
\includegraphics[width=\columnwidth]{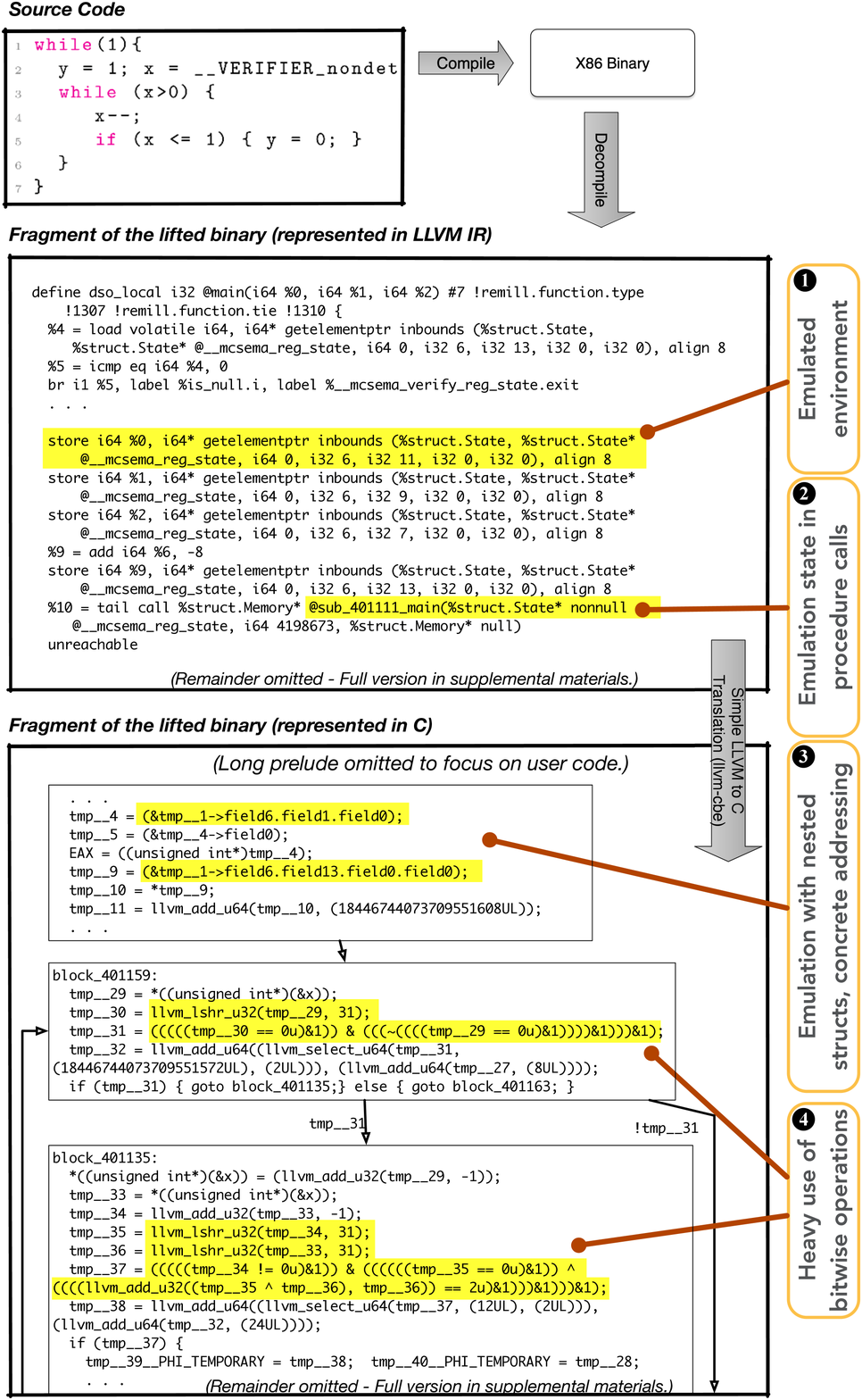}
\caption{\label{fig:rawcfg} Challenges involved in reasoning about the lifted binary of \code{PotentialMinimizeSEVPABug.c}.}
\end{figure}

We here explain via an example the need for the translations
perform by \Tool{}, discussed in Sec.~\ref{sec:casestudy}.
Figure~\ref{fig:rawcfg} shows an (condensed) example of the result of lifting a GCC-compiled binary version of \code{PotentialMinimizeSEVPABug.c} (using \mcsema). 
In the following, we describe how decompilation tools (\eg~\mcsema) target \emph{re-compilation} and the challenges this poses for  existing verification techniques and tools. We then describe translations performed by \Tool\ to make lifted
binaries more amenable to verification. As discussed in Sec.~\ref{sec:casestudy}, the translations below were implemented as passes on the lifted LLVM IR.

\noindent
{\bf Run-time environment.} Binary lifting de-compiles into a program that 
mimics the  binary behavior.
To ensure that  a new \emph{re}-compiled binary would run correctly, lifting yields code that switches the contexts between the run-time environments and the simulated code, 
somewhat akin to how a loader first moves environment variables onto the stack. 
This context-switch code wraps around the simulated program and is indispensable for execution of re-compiled code.  
Context-switch code can be fairly complex. 
For example, it frequently uses assembly code to move values between the physical environment (\eg\ environment variables like {\tt \$PATH}) and the simulated constructs (\eg\ registers and the stack). 
The initial state of the registers is also loaded from the runtime, 
as seen by the instructions in callout \cirC{1} in Fig.~\ref{fig:rawcfg}.
Tracing back the origins of these values, which stem from the runtime code,
poses a nearly impossible, yet unnecessary task for verification tools.
For most verification tasks, it does not matter where the initial values of registers or environment  variables such as {\tt \$PATH} come from; we can just treat them as nondeterministic input values on the stack.

\emph{Translation: Removing verification-unrelated code and data.} 
We implemented a pass to analyze lifted output and decouple
context-switch code from the code that simulates the 
original program. We first locate the original 
{\tt main} function in the simulated code and then follow the 
control flow to statically extract and trim code that can reach {\tt main}. 
Further, the context-switch code also includes program-dependent functions 
that are registered to be executed before the {\tt main} or after 
the {\tt exit}. To avoid missing such functions in verification, 
we allocate calls to them at the begin or the end of the {\tt main} 
function, following the order these functions are 
registered to run in the original binary. 

\noindent
{\bf Passing emulation state through procedures.} Binary lifting for recompilation  generates programs in which function calls are used to pass emulation state.
This can be seen, for example, in callout \cirC{2} in Fig.~\ref{fig:rawcfg},
where \code{struct.State} is passed as the first argument to \code{sub\_401111\_main}.
These arguments are \emph{not} part of the original program. Rather, lifting introduces these additional arguments to simulate the possibility of context-switches within function calls. 
Further, when the lifted code is recompiled, code simulating the callee can access the contexts through these \code{struct.State} arguments, which is typically more efficient than directly accessing the global data structure.
%
Interprocedural reasoning is known to make verification more challenging and requires more sophisticated algorithms such as procedure summaries~\cite{summaries,dillig2011precise} and nested interpolants~\cite{nestedInterpolants}, especially when context sensitivity is required. When machine emulation involves expanding use of arguments, we found this complicates analysis and hampers verification.

\emph{Translation: Simplifying function arguments.} 
Adding arguments to user procedures (as done by \mcsema) complicates verification.
Fortunately, a translation is possible: the arguments in a \mcsema{}-generated function point to 
the same global data structure.
As such, we eliminate these arguments from every function call. We then create a pointer pointing to the global data struct
and replace all \emph{uses} of the first argument in the function body with 
uses of our new pointer.

\noindent
{\bf Nested structures for emulation.}
Lifted binaries encode complicated structures that simulate hardware features such as registers, arithmetic flags, FPU status flags, the stack, and instruction pointers. 
These are represented as nested structures, \eg
\code{state->general\_registers.register13.union.uint64cell}.
This can be seen, for example, in callout \cirC{3} in Fig.~\ref{fig:rawcfg}, where 
the field \code{tmp\_\_1->field6.field13.field0.field0} is accessed.
The use of nesting in these structures provides efficiency: constructs that are commonly used together (\eg\ general purpose registers) can be artificially grouped to the same cache line,  avoiding cache evictions.
However, reasoning about these nested data-structures is difficult because verification tools cannot make any assumptions about where these data-structures come from, how they are used and, most importantly, how they may be aliased. 
Consequently, verification tools have to carefully track heap references to infer non-aliasing, even though the lifting process ensures that they will not.

\emph{Translation: Flattening the emulation state.} 
Many of the data structures for the emulated state are functionally independent and hence the complex nesting is not necessary to maintain the original semantics. We implemented a pass to flatten the data structures. We create individual variables for all the innermost and separable fields. We then translate accesses to these nested structures, with use-define reasoning to identify all the accesses to a flattened field. For the aforementioned {\tt state->}{\tt general\_registers.}{\tt register13.}{\tt union.}{\tt uint64cell}, we allocate a new global variable {\tt register13} with the same type of {\tt uint64cell} and re-locate all the original accesses to {\tt register13}.

\def\por{\enspace|\enspace}
\def\pdef{\Coloneqq}

\def\id{\mathit{Id}}

\def\idbit{\mathit{Id'}}
\def\lit{\mathit{Lit}}
\def\const{\mathit{Const}}

\def\pdec{\mathit{DecIR}}

\def\stmt{\mathit{Stmt}}
\def\astmt{\mathit{AStmt}}
\def\lhs{\mathit{Lhs}}
\def\func{\mathit{Func}}


\def\labe{\mathit{Label}}

\def\nexpr{\mathit{NondetExpr}}
\def\expr{\mathit{Expr}}
\def\binexpr{\mathit{BinExpr}}
\def\bitexpr{\mathit{BitExpr}}
\def\lhs{\mathit{Lhs}}

\def\ariop{\mathit{ArithOp}}
\def\logop{\mathit{LogicOp}}
\def\bitop{\mathit{BitOp}}
\def\ubitop{\mathit{UbitOp}}
\def\bitconst{\mathit{BitConst}}

\def\plus#1{#1(\lstinline[style=lvir]$,$~#1)^*}


%

%

\section{Detailed Results for \Tool{} Translation}
\label{apx:fabe}

Table \ref{tab:lift-term-detail-mcsema} and Table \ref{tab:lift-term-detail-simplify} show the details of applying termination verifiers to \mcsema{} output and the output after \Tool's translations, respectively.

\begin{table}
    \centering
    \caption{\label{tab:lift-term-detail-mcsema} Details for termination verification of vanilla \mcsema\  binary lifting.}
    \resizebox{\textwidth}{!}{\begin{tabular}{lcrcrcrcrcrcrc}
    \toprule
                                        &          & \multicolumn{2}{c}{\aprove} & \multicolumn{2}{c}{\cpachecker} & \multicolumn{2}{c}{\kittel} & \multicolumn{2}{c}{\twols} & \multicolumn{2}{c}{\ultimate} & \multicolumn{2}{c}{\Tool}                                                        \\
    Benchmark                           & Expected & Time                        & Result                          & Time                        & Result                     & Time                          & Result                    & Time  & Result & Time    & Result & Time    & Result \\
    \cmidrule(r){1-1}\cmidrule(lr){2-2}\cmidrule(r){3-4}\cmidrule(lr){5-6}\cmidrule(lr){7-8}\cmidrule(lr){9-10}\cmidrule(lr){11-12}\cmidrule(l){13-14}
    Singapore-2\_gccO0.mcsema.cbe.c     &          & 8.28s                       & \UNK                            & 2.04s                       & \FAIL                      & 0.06s                         & \FAIL                     & 0.17s & \UNK   & 900.43s & \TOUT  & 900.36s & \TOUT  \\
    aaron2-2\_gccO0.mcsema.cbe.c        &          & 5.48s                       & \UNK                            & 2.13s                       & \FAIL                      & 0.06s                         & \FAIL                     & 0.17s & \UNK   & 900.37s & \TOUT  & 900.43s & \TOUT  \\
    Singapore\_plus\_gccO0.mcsema.cbe.c &          & 4.13s                       & \UNK                            & 2.58s                       & \FAIL                      & 0.06s                         & \FAIL                     & 0.17s & \UNK   & 900.45s & \TOUT  & 900.56s & \TOUT  \\
    Mysore-1\_gccO0.mcsema.cbe.c        &          & 3.77s                       & \UNK                            & 1.98s                       & \FAIL                      & 0.05s                         & \FAIL                     & 0.16s & \UNK   & 900.35s & \TOUT  & 900.39s & \TOUT  \\
    Parallel\_gccO0.mcsema.cbe.c        &          & 3.01s                       & \UNK                            & 2.12s                       & \FAIL                      & 0.06s                         & \FAIL                     & 0.15s & \UNK   & 900.38s & \TOUT  & 900.33s & \TOUT  \\
    Pure2Phase-1\_gccO0.mcsema.cbe.c    &          & 8.60s                       & \UNK                            & 2.50s                       & \FAIL                      & 0.05s                         & \FAIL                     & 0.16s & \UNK   & 900.38s & \TOUT  & 900.44s & \TOUT  \\
    Thun-1\_gccO0.mcsema.cbe.c          &          & 3.55s                       & \UNK                            & 2.41s                       & \FAIL                      & 0.05s                         & \FAIL                     & 0.15s & \UNK   & 900.44s & \TOUT  & 900.42s & \TOUT  \\
    easy2-2\_gccO0.mcsema.cbe.c         &          & 2.72s                       & \UNK                            & 2.22s                       & \FAIL                      & 0.05s                         & \FAIL                     & 0.16s & \UNK   & 589.64s & \FAIL  & 760.59s & \OOM   \\
    aaron3-2\_gccO0.mcsema.cbe.c        &          & 9.82s                       & \UNK                            & 2.44s                       & \FAIL                      & 0.06s                         & \FAIL                     & 0.20s & \UNK   & 900.41s & \TOUT  & 900.47s & \TOUT  \\
    Pure3Phase-2\_gccO0.mcsema.cbe.c    &          & 5.85s                       & \UNK                            & 2.54s                       & \FAIL                      & 0.06s                         & \FAIL                     & 0.22s & \UNK   & 900.40s & \TOUT  & 900.43s & \TOUT  \\
    easy\_debug\_gccO0.mcsema.cbe.c     &          & 3.69s                       & \UNK                            & 2.28s                       & \FAIL                      & 0.06s                         & \FAIL                     & 0.19s & \UNK   & 900.40s & \TOUT  & 902.27s & \TOUT  \\
    Mysore-2\_gccO0.mcsema.cbe.c        &          & 7.70s                       & \UNK                            & 2.03s                       & \FAIL                      & 0.06s                         & \FAIL                     & 0.17s & \UNK   & 900.43s & \TOUT  & 900.49s & \TOUT  \\
    easy1\_gccO0.mcsema.cbe.c           &          & 2.92s                       & \UNK                            & 2.24s                       & \FAIL                      & 0.05s                         & \FAIL                     & 0.16s & \UNK   & 785.05s & \FAIL  & 802.15s & \OOM   \\
    aaron2-1\_gccO0.mcsema.cbe.c        &          & 9.27s                       & \UNK                            & 2.15s                       & \FAIL                      & 0.05s                         & \FAIL                     & 0.16s & \UNK   & 900.31s & \TOUT  & 900.37s & \TOUT  \\
    easy2-1\_gccO0.mcsema.cbe.c         &          & 2.60s                       & \UNK                            & 1.83s                       & \FAIL                      & 0.02s                         & \FAIL                     & 0.11s & \UNK   & 566.99s & \FAIL  & 870.64s & \OOM   \\
    Thun-2\_gccO0.mcsema.cbe.c          &          & 5.91s                       & \UNK                            & 2.41s                       & \FAIL                      & 0.03s                         & \FAIL                     & 0.16s & \UNK   & 900.36s & \TOUT  & 900.45s & \TOUT  \\
    Pure2Phase-2\_gccO0.mcsema.cbe.c    &          & 4.57s                       & \UNK                            & 1.90s                       & \FAIL                      & 0.02s                         & \FAIL                     & 0.13s & \UNK   & 900.38s & \TOUT  & 900.35s & \TOUT  \\
    aaron3-1\_gccO0.mcsema.cbe.c        &          & 46.87s                      & \UNK                            & 1.91s                       & \FAIL                      & 0.03s                         & \FAIL                     & 0.13s & \UNK   & 900.42s & \TOUT  & 900.44s & \TOUT  \\
    \bottomrule
\end{tabular}
}
\end{table}

\begin{table}
    \centering
    \caption{\label{tab:lift-term-detail-simplify} Details for termination verification of \Tool\ translated lifted binaries.}
    \resizebox{\textwidth}{!}{\begin{tabular}{lcrcrcrcrcrcrc}
    \toprule
                                                  &          & \multicolumn{2}{c}{\aprove} & \multicolumn{2}{c}{\cpachecker} & \multicolumn{2}{c}{\kittel} & \multicolumn{2}{c}{\twols} & \multicolumn{2}{c}{\ultimate} & \multicolumn{2}{c}{\Tool}                                                     \\
    Benchmark                                     & Expected & Time                        & Result                          & Time                        & Result                     & Time                          & Result                    & Time  & Result & Time  & Result & Time   & Result \\
    \cmidrule(r){1-1}\cmidrule(lr){2-2}\cmidrule(r){3-4}\cmidrule(lr){5-6}\cmidrule(lr){7-8}\cmidrule(lr){9-10}\cmidrule(lr){11-12}\cmidrule(l){13-14}

    Parallel\_gccO0.simplify.cbe.c.instr.c        &          & 1.85s                       & \UNK                            & 900.51s                     & \TOUT                      & 0.01s                         & \FAIL                     & 0.12s & \UNK   & 8.14s & \OK    & 7.27s  & \OK    \\
    Thun-1\_gccO0.simplify.cbe.c.instr.c          &          & 1.63s                       & \UNK                            & 900.55s                     & \TOUT                      & 0.01s                         & \FAIL                     & 0.15s & \UNK   & 9.68s & \OK    & 9.95s  & \OK    \\
    aaron2-1\_gccO0.simplify.cbe.c.instr.c        &          & 1.43s                       & \UNK                            & 900.48s                     & \TOUT                      & 0.02s                         & \FAIL                     & 0.12s & \UNK   & 7.91s & \OK    & 10.32s & \OK    \\
    Pure3Phase-2\_gccO0.simplify.cbe.c.instr.c    &          & 1.86s                       & \UNK                            & 900.43s                     & \TOUT                      & 0.02s                         & \FAIL                     & 0.12s & \UNK   & 8.34s & \OK    & 9.44s  & \OK    \\
    Mysore-2\_gccO0.simplify.cbe.c.instr.c        &          & 1.40s                       & \UNK                            & 900.38s                     & \TOUT                      & 0.02s                         & \FAIL                     & 0.15s & \UNK   & 7.03s & \OK    & 10.25s & \OK    \\
    easy1\_gccO0.simplify.cbe.c.instr.c           &          & 1.92s                       & \UNK                            & 900.29s                     & \TOUT                      & 0.02s                         & \FAIL                     & 0.11s & \UNK   & 6.47s & \OK    & 7.79s  & \OK    \\
    Pure2Phase-1\_gccO0.simplify.cbe.c.instr.c    &          & 1.58s                       & \UNK                            & 900.35s                     & \TOUT                      & 0.02s                         & \FAIL                     & 0.12s & \UNK   & 7.19s & \OK    & 8.00s  & \OK    \\
    aaron3-2\_gccO0.simplify.cbe.c.instr.c        &          & 1.88s                       & \UNK                            & 900.40s                     & \TOUT                      & 0.02s                         & \FAIL                     & 0.12s & \UNK   & 7.22s & \OK    & 7.54s  & \OK    \\
    easy2-1\_gccO0.simplify.cbe.c.instr.c         &          & 1.63s                       & \UNK                            & 900.44s                     & \TOUT                      & 0.01s                         & \FAIL                     & 0.12s & \UNK   & 6.73s & \OK    & 6.81s  & \OK    \\
    aaron3-1\_gccO0.simplify.cbe.c.instr.c        &          & 1.46s                       & \UNK                            & 900.45s                     & \TOUT                      & 0.02s                         & \FAIL                     & 0.14s & \UNK   & 8.09s & \OK    & 12.30s & \OK    \\
    Pure2Phase-2\_gccO0.simplify.cbe.c.instr.c    &          & 1.87s                       & \UNK                            & 900.48s                     & \TOUT                      & 0.02s                         & \FAIL                     & 0.09s & \UNK   & 8.09s & \OK    & 6.96s  & \OK    \\
    easy2-2\_gccO0.simplify.cbe.c.instr.c         &          & 1.79s                       & \UNK                            & 900.47s                     & \TOUT                      & 0.02s                         & \FAIL                     & 0.09s & \UNK   & 8.36s & \OK    & 6.40s  & \OK    \\
    Mysore-1\_gccO0.simplify.cbe.c.instr.c        &          & 1.95s                       & \UNK                            & 900.41s                     & \TOUT                      & 0.01s                         & \FAIL                     & 0.10s & \UNK   & 6.55s & \OK    & 8.63s  & \OK    \\
    aaron2-2\_gccO0.simplify.cbe.c.instr.c        &          & 1.48s                       & \UNK                            & 900.55s                     & \TOUT                      & 0.01s                         & \FAIL                     & 0.11s & \UNK   & 8.19s & \OK    & 8.87s  & \OK    \\
    easy\_debug\_gccO0.simplify.cbe.c.instr.c     &          & 1.52s                       & \UNK                            & 900.54s                     & \TOUT                      & 0.02s                         & \FAIL                     & 0.10s & \UNK   & 8.16s & \OK    & 9.27s  & \OK    \\
    Thun-2\_gccO0.simplify.cbe.c.instr.c          &          & 1.88s                       & \UNK                            & 900.55s                     & \TOUT                      & 0.02s                         & \FAIL                     & 0.10s & \UNK   & 7.02s & \OK    & 7.71s  & \OK    \\
    Singapore\_plus\_gccO0.simplify.cbe.c.instr.c &          & 1.83s                       & \UNK                            & 900.56s                     & \TOUT                      & 0.01s                         & \FAIL                     & 0.15s & \UNK   & 6.96s & \OK    & 6.95s  & \OK    \\
    Singapore-2\_gccO0.simplify.cbe.c.instr.c     &          & 1.45s                       & \UNK                            & 900.51s                     & \TOUT                      & 0.02s                         & \FAIL                     & 0.10s & \UNK   & 7.05s & \OK    & 7.68s  & \OK    \\
    \bottomrule
\end{tabular}}
\end{table}

\section{Detailed Results for LTL of Lifted Binaries}
\label{apx:details:ltl:lift}

\begin{table}[t]
  \caption{\label{tab:lift-ltl-details} Details for LTL lifted binary benchmarks, using vanilla \mcsema{} versus \Tool's translated IR and vanilla \ultimate{} versus \Tool{}'s bitwise-branching (Section~\ref{sec:bitwise}). Gray cells are unsound, green cells use slightly different settings (enabled SBE).}
  \centering
  \resizebox{\textwidth}{!}{\begin{tabular}{llcc@{\hspace{1em}}rcrcc@{\hspace{1em}}rcrc}
    \toprule

                            &                                              &      &  & \multicolumn{4}{c}{Vanilla \mcsema{} IR}    &                           & \multicolumn{4}{c}{\Tool's translated IR}                                                                                                     \\
    \cmidrule(r){5-8}\cmidrule(l){10-13}
                            &                                              &      &  & \multicolumn{2}{c}{\ultimate} & \multicolumn{2}{c}{\Tool} &                              & \multicolumn{2}{c}{\ultimate} & \multicolumn{2}{c}{\Tool}                                         \\

    Benchmark               & Property                                     & Exp. &                               & Time                      & Result                       & Time                          & Result &      & Time       & Result   & Time       & Result\\
    \cmidrule(r){1-1}\cmidrule(lr){2-2}\cmidrule(l){3-3} \cmidrule(r){5-6}\cmidrule(l){7-8} \cmidrule(r){10-11}\cmidrule(l){12-13}
    01-exsec2.s.c           & $\lozenge (\square x=1)$                     & \OK  &  & 4.45s                         & \FAIL                     & 4.56s                        & \FAIL                         &                           & 11.01s & \OK  & 11.23s     & \OK      \\
    01-exsec2.s.f.c.c       & $\lozenge (\square x \neq 1)$                & \NOK &  & 6.31s                         & \FAIL                     & 5.79s                        & \FAIL                         &                           & 9.21s  & \UNK & 10.36s     & \NOK     \\
    SEVPA\_gccO0.s.c        & $\square  ( x > 0 \Rightarrow \lozenge y=0)$ & \OK  &  & 6.31s                         & \FAIL                     & 5.93s                        & \FAIL                         &                           & 11.17s & \UNK & \hlg22.92s & \hlg\OK  \\
    SEVPA\_gccO0.s.f.c      & $\square  ( x > 0 \Rightarrow \lozenge y=2)$ & \NOK &  & 5.16s                         & \UNK                      & 5.25s                        & \UNK                          &                           & 35.46s & \UNK & \hlg14.92s & \hlg\NOK \\
    acqrel.simplify.s.c     & $\square (x=0 \Rightarrow \lozenge y=0)$     & \OK  &  & 5.17s                         & \FAIL                     & 5.38s                        & \FAIL                         &                           & 10.66s & \OK  & 9.00s      & \OK      \\
    acqrel.simplify.s.f.c.c & $\square (x=0 \Rightarrow \lozenge y=1)$     & \NOK &  & 6.06s                         & \FAIL                     & 5.48s                        & \FAIL                         &                           & 21.12s & \UNK & 17.60s     & \NOK     \\
    example1\_fea.s.c       & $\square error = 0$                     & \NOK &  & \hl 13.06s                    & \hl\OK                    & \hl13.44s                    & \hl\OK                        &                           & 6.96s  & \NOK & 6.25s      & \NOK     \\
    example1\_sea.s.c       & $\square error = 0$            & \OK  &  & 11.22s                        & \OK                       & 11.11s                       & \OK                           &                           & 12.02s & \OK  & 8.51s      & \OK      \\
    example2\_fea.s.c       & $\square error = 0$            & \NOK &  & \hl 12.26s                    & \hl \OK                   & \hl13.54s                    & \hl\OK                        &                           & 10.99s & \UNK & 12.66s     & \NOK     \\
    example2\_sea.s.c       & $\square error = 0$            & \OK  &  & 14.24s                        & \OK                       & 14.15s                       & \OK                           &                           & 11.75s & \UNK & 8.94s      & \OK      \\
    example3\_fea.s.c       & $\square error = 0$            & \NOK &  & \hl 10.47s                    & \hl \OK                   & \hl12.02s                    & \hl \OK                       &                           & 8.51s  & \UNK & 8.34s      & \NOK     \\
    example3\_sea.s.c       & $\square error = 0$            & \OK  &  & 11.11s                        & \OK                       & 11.36s                       & \OK                           &                           & 8.70s  & \UNK & 6.67s      & \OK      \\
    exsec2.simplify.s.c     & $\square \lozenge x=1$                       & \OK  &  & 4.92s                         & \FAIL                     & 4.96s                        & \FAIL                         &                           & 6.38s  & \OK  & 5.60s      & \OK      \\
    exsec2.simplify.s.f.c.c & $\square \lozenge x\neq 1$                   & \NOK &  & 4.55s                         & \FAIL                     & 5.22s                        & \FAIL                         &                           & 7.85s  & \NOK & 6.28s      & \NOK     \\
    nondet\_gccO0.s.c       & $\square x > 0$                              & \NOK &  & 4.57s                         & \FAIL                     & 4.92s                        & \FAIL                         &                           & 10.99s & \UNK & 10.59s     & \NOK     \\
    simple3\_gccO0.s.c      & $\lozenge p=1$                               & \OK  &  & 4.97s                         & \FAIL                     & 5.82s                        & \FAIL                         &                           & 91.22s & \UNK & 15.87s     & \OK      \\
    simple3\_gccO0.s.f.c    & $\lozenge p=2$                               & \NOK &  & 4.96s                         & \FAIL                     & 4.87s                        & \FAIL                         &                           & 80.77s & \UNK & \hlg20.85s & \hlg\NOK \\
    \bottomrule
\end{tabular}
}
\end{table}

Table~\ref{tab:lift-ltl-details} is a more detailed version of the table from Sec.~\ref{sec:casestudy}, this time comparing the performance of
\ultimate\ versus \Tool, when applied first to vanilla \mcsema{} IR, and then applied to \Tool's translated IR.
The experimental result shows that our translations significantly eliminates all crashes and possibly unsound results occurring in the verification of the original lifted code. The unsoundness came from the fact that {\ultimate} detected possible memory errors in the code snippet setting run-time environment up in those programs, thus considering the programs to be infeasible with respect to the LTL properties and assuming that they always hold. {\Tool} inherits such behavior from {\ultimate}. Moreover, the results on the simplified lifted code 
These results highlights the effectiveness of our bitwise branching technique which helps {\Tool} to prove the LTL properties of all 17 benchmark correctly while {\ultimate} can only prove 6 of them. 
The possible memory errors were eliminated when our translation flattened the emulation state.

\end{document}